\documentclass[preprint,3p,10pt]{elsarticle}

\usepackage{amssymb}
\usepackage{graphicx}
\usepackage{amsmath}
\usepackage{booktabs}
\newcommand{\sD}{{\sf D}}

\usepackage{graphicx}
\usepackage{latexsym}
\usepackage{amsthm}
\usepackage{amsmath}
\usepackage{bm}
\usepackage{amsbsy}
\usepackage{grffile}
\usepackage{ifpdf}
\usepackage{epstopdf}
\usepackage{longtable}
\usepackage{color}
\usepackage{graphicx}
\usepackage{caption}
\usepackage{subcaption}
\usepackage{multicol}
\usepackage{lscape}

\newcommand{\tr}{\text{tr}}
\newcommand{\vect}{\text{$\mathrm{vec}$}}

\newcommand{\df}{\text{df}}
\newcommand{\GIC}{\text{GIC}}
\newcommand{\KL}{\text{KL}}

\newcommand{\omegahat}{\hat{\omega}}
\newcommand{\Omegahat}{\hat{\Omega}}

\journal{arXiv}
\newtheorem{lemma}{Lemma}

\begin{document}

\begin{frontmatter}



 \title{Generalized information criterion for model selection in penalized graphical models}
 \author[label1]{A. Abbruzzo\corref{cor1}\fnref{labels}}
  \fntext[labels]{Corresponding author: Antonino Abbruzzo\\
   Phone number: 3293883847}
  \ead{antonino.abbruzzo@unipa.it}
 \author[label2]{Ivan Vuja\v{c}i\'c}
 \author[label2]{Ernst Wit}
 \author[label1]{Angelo M. Mineo}
 \address[label1]{Dipartimento Scienze Economiche, Aziendali e Statistiche, University of Palermo, Viale delle Scienze Ed. 13, 90128 Palermo, Italy }
 \address[label2]{Johann Bernoulli Institute, University of Groningen, Nijenborgh 9, 9747 AG Groningen, The Netherlands}

\begin{abstract}
This paper introduces an estimator of the relative directed distance between an estimated model and the true model, based on the Kulback-Leibler divergence and is motivated by the generalized information criterion proposed by Konishi
and Kitagawa. This estimator can be used to select model in penalized Gaussian copula graphical models. The use of this estimator is not feasible for high-dimensional cases. However, we derive an efficient way to compute this 
estimator which is feasible for the latter class of problems. Moreover, this estimator is, generally, appropriate for several penalties such as lasso, adaptive lasso and smoothly clipped absolute deviation penalty.
Simulations show that the method performs similarly to KL oracle estimator and it also improves BIC performance in terms of support recovery of the graph. Specifically, we compare our method with Akaike information criterion, Bayesian information criterion and cross validation for band, sparse and dense network structures. 
\end{abstract}


\end{frontmatter}



\section{Introduction}
Graphical models are powerful tools for analyzing relationships between a large number of random variables. Their fundamental importance and universal applicability are due to two main features. Firstly, graphs can be used to represent
 complex dependences among random variables. Secondly, their structure is modular which means that complex graphs can be built from many simpler graphs.  A graph consists of a set of nodes and a set of links between these nodes. In a
 graphical model nodes are associated with random variables and links represent conditional dependencies \cite{lauritzen96}. Many modern experimental studies consist of observing many features across a small amount of statistical units. For example in genetics, variables
are typically thousands of genes or gene products, whereas there are only few hundreds of observations. At the same time, many of these variables are pairwise conditionally independent from each other, when we condition on the rest of the variables. This could be because a lot of variables encode effectively the same information, as might be the case in financial time series. Or because there are specific functional relationships between the variables, as is the case in genomic pathways. Either way, this means that many of the underlying graphs in these experiments are sparse. 

Here, we focus on the class of Gaussian copula graphical models \cite{xue2000multivariate}. In such models, nodes can have arbitrary marginal distributions, but have a correlation structure that is induced by means of a multivariate Gaussian copula.
 This model has as a special case the  Gaussian graphical model (GGM). In such model, the random variables follow a multivariate normal distribution. An important property of GGMs is that the precision matrix determines the conditional independence graph. As the dependence structure in a Gaussian copula model is determined uniquely by the underlying normal distribution,  also the conditional dependence graph for Gaussian copula models is determined uniquely by the precision matrix.

Several methodologies have been proposed for estimating the precision matrix associated with an unknown Gaussian graphical model. \citet{whittaker2009graphical} describes likelihood ratio tests. Recently attention has been 
focused on penalized likelihood approaches. The idea is to penalize the maximum likelihood function with a sparsity penalty on the precision matrix. A tuning parameter in the penalty regulates the sparsity: the larger
the tuning parameter, the more  zeroes will be estimated in the precision matrix. In this setting, \citet{meinshausen2006high} proposed to select edges for each node in the graph by regressing the variable on all other 
variables using $\ell_1$ penalized regression. Penalized maximum likelihood approaches using the $\ell_1$-norm have been considered by \citet{yuan2007model,banerjee2008model,rothman2008sparse}, who have all proposed
 different algorithms for computing an estimator of the precision matrix. An efficient graphical lasso algorithm, glasso, was proposed by \citet{friedman2008sparse}. \citet{lam2009sparsistency} study properties of general penalties that include lasso and adaptive lasso and SCAD as a special cases. 
\par
A crucial issue with penalized maximum likelihood estimation is the choice of the regularization parameter, which controls the amount of penalization. We can divide the methods for choosing the regularization parameter in
various ways. One type of methods are computation-based methods such as cross-validation (CV) \citep{rothman2008sparse,fan2009network,schmidt2010graphical,fitch2012computationally}. The other type are information criteria such as AIC \citep{menendez2010gene}, 
BIC \citep{yuan2007model,schmidt2010graphical,menendez2010gene,lian2011shrinkage,gao2012tuning} and EBIC \citep{foygel2010ebic}.
In principle, computation-based methods can be expected to be more accurate since, usually, they involve reusing the data, but they are computationally more expensive. Information criteria are less accurate
 but computationally faster. On the other hand, one can divide the methods in those that try to maximize the prediction power, i.e. minimize Kulback-Leibler divergence, such as the AIC and CV, and those whose aim is model selection consistency, such as BIC and EBIC. 

In this paper, we introduce the Gaussian copula graphical models (CGGMs) based on penalized likelihood approaches
(Section \ref{sec:graphicalmodel}). Then, motivated by Konishi's generalized information criterion \citep{Konishi96,Konishi08}, we derive an estimator of the relative distance between
a model and the true model, based on the Kulback-Leibler divergence
(Section \ref{sec:gic}). This estimator can be used as model selection criterion for GCGMs. The direct use of the estimator derived in Section \ref{sec:gic} would make it not feasible for high-dimensional problems. So, we derive a feasible formula to compute the proposed estimator efficiently (Section \ref{sec:egic}). In Section \ref{sec:sim}, we show a simulation study. Simulation study shows, that the method approximates the Kulback-Leibler divergence much better than the AIC and that in combination with
other selection methods it can be used to estimate graph structures from data.




\section{Graphical models based on penalized likelihood}
\label{sec:graphicalmodel}

Let  $X=(X^{(1)}, \ldots, X^{(d)})^{\top}$ be a random vector with probability distribution $P$. The main advantage of using graphical models is that conditional independence relationships among $\{X^{(i)}: i = 1,\ldots,d\}$ can be summarized in an undirected graph $G=(V,E)$ where $V = \{1, \ldots, d\}$ is a set of nodes.
We associate with each node $i$ the random variable $X^{(i)}$. The set $E$ is the set of undirected edges, where an undirected edge $e=\{i,j\} \subset V$ is a subset with two elements of $V$. The absence of an edge between nodes $i$ and $j$ corresponds with conditional independence of the random variables $X^{(i)}$ and $X^{(j)}$ given the rest. The pair $(G,P)$ is referred to as a graphical model, if 
\[ X^{(i)} \perp_P X^{(j)} | X^{V\setminus\{i,j\}} \Longleftrightarrow (i,j) \not \in E ,\]
i.e. when the links represent conditional dependence relationships between the random variables. In this paper we consider only models with strictly positive densities and therefore all the various Markov properties are equivalent and the density permits a factorization according to the clique structure in the graph \cite{lauritzen96}.

Assume $n$ i.i.d observations where $X_k \sim N(0, \Omega^{-1}), k = 1, \ldots, n$, i.e. $P$ is multivariate normal density function, then $(G, P)$ is referred to Gaussian graphical model. If an element of the precision matrix $\Omega$, denoted by $\omega_{i,j}$, is equal to zero, then it can be proved that $X^{(i)}$ and $X^{(j)}$ are conditionally independent given the rest \cite{whittaker2009graphical}. Thus the parameter estimation and model selection for Gaussian graphical models are equivalent to estimating parameters and identifying zeros in the precision matrix. 
 
Suppose we have $n$ multivariate observations $x_1, \ldots, x_n$ of dimension $d$ from distribution $N(0,\Omega_0^{-1})$. The log likelihood function is up to an additive constant equal to:

\begin{equation}
\label{eq:gaussian}
l(\Omega)= \frac{n}{2}\left\{\mbox{log}|\Omega| - \mbox{tr}\left(\mbox{S}\Omega\right) \right\},
\end{equation}
where $\mbox{S}=(1/n) \sum_{k=1}^n x_k x_k^{\top}$,	$|\cdot|$ denotes the determinant, and $\mbox{tr}(\cdot)$ is the trace.

In order to have a more flexible family of models so that multivariate normal assumption can be relaxed 
we make use of the Gaussian copula model. 

\subsection{Gaussian copula graphical models}
The multivariate normal assumption for the joint distribution of $X$ can be relaxed if one assumes that the multivariate dependence patterns among the observed variables $X$ are given by the Gaussian copula. The class of graphical models given by $(G, P)$, where $P$ is a Gaussian copula density, is called  Gaussian copula graphical models. These models can have arbitrary marginal distributions and have a dependence structure induced via a multivariate Gaussian copula, i.e.:

\begin{equation}
C(u_{1} , \ldots, u_{d} |\Gamma) =  \Phi(\phi^{-1}(u_{1}), \ldots , \phi^{-1}(u_{d}) | \Gamma),
\end{equation}
where $\phi$ is the CDF of the standard normal distribution, $u_i = F_i(x_i)$ with $F_i$ a univariate distribution function, and $ \Phi(\cdot|\Gamma)$ is the CDF of a multivariate normal distribution with correlation matrix $\Gamma$. Generally, $F_i$ can be arbitrary marginal distributions. However, following arguments given in \cite{hoff2007extending}, we treat the marginal distributions $F_i$ as nuisance parameters. So,
we estimate them non parametrically since we are mainly interested in the dependences among the random variables. 

Given $n$ i.i.d. observations from a Gaussian copula, the corresponding log density function can be written as follows (see \cite{xue2000multivariate}):

\begin{equation}
c(u_1, \ldots, u_d;\Gamma) = -\frac{n}{2} \left\{\mbox{log}| \Gamma| - \mbox{tr}(\tilde{\mbox{S}} \Gamma^{-1}) + q^{\top}q 
\right\},
\label{eq:copula2}
\end{equation}
where $\tilde{\mbox{S}} =  \sum_{k=1}^n q_k q_k^{\top}$, with $q_i = \phi^{-1}(u_i)$, $i = 1, \ldots, d$,
and $q = (\phi^{-1}(u_1),  \ldots, \phi^{-1}(u_d))$. It can be shown that 
conditional independences between two random variables $i$ and $j$ given the rest are given by zeroes elements in $\Omega = \Gamma^{-1}$ \cite{dobra2009copula}, i.e. :
\begin{equation}
\label{eq:thetais0} \{i,j\} \notin E \Longleftrightarrow X^{(i)} \perp X^{(j)}|X^{(V\setminus\{i,j\})}\Longleftrightarrow \Omega_{ij} = 0.
\end{equation}
Alternatively, the Gaussian copula can be constructed by introducing a vector of latent variables $Z = (Z_1, \ldots ,Z_d)$ that are related to the observed variables $X$ as $X_j = F_j^{-1}(\Phi(Z_j))$,
where $F_j^{-1}$ is the pseudo inverse of $F_j$. As argued in \cite{dobra2009copula}, the Markov properties associated with a graphical model for the latent variables hold for the observed variables if all the marginals $\{F_j: j = 1, \ldots, d\}$ are continuous. The presence of some discrete marginals might induce additional dependencies that are not modeled in a graphical model with the latent variables. However, in what follows we consider continuous random variables .  The GGM is a special case of the Gaussian copula graphical models. In this case the marginal distributions $F_j$ are assumed normal. 



The set of all possible graphs is $2^{d(d-1)/2}$ which is huge even for relatively small $d$. However, sparsity of the graph is a main assumption in most application of graphical models in genetics. Penalized likelihood approaches can be used in order to incorporate the sparsity assumption and make estimation procedures feasible even for large $d$. Lasso, adaptive lasso and SCAD are three most commonly used in penalized likelihood approaches. In the next section, we briefly recall these penalty functions in conjunction with their properties.

\subsection{Penalty functions for Gaussian copula graphical models}
\label{sec:penalizedinference}
The aim of penalized inference in graphical models is to induce a zero structure in the matrix $\Omega$, which according to (\ref{eq:thetais0}) results in a particular edge structure $E$ of the graphical model $(G,P_{\Omega})$. The
penalized likelihood is obtained by adding a penalty $p_{\lambda_{ij} }$ function  to the likelihood. For a variety of penalty functions, the resulting maximum penalized likelihood estimate (MPLE) is as follows:
\begin{equation}
\label{eq:mple}
\hat{\Omega}_\lambda = \underset{\Omega\succ 0}{\operatorname{argmax}} \mbox{ log}|\Omega |- \mbox{tr}(
\tilde{\mbox{S}} \Omega) - \sum_{i\neq j}^p  p_{\lambda_{ij}}(|\omega_{ij}|)
\end{equation}
$p_{\lambda_{ij} }$ depends on a tuning parameter $\lambda_{ij}$ which regulates the weights on the corresponding $\omega_{ij}$ terms of the precision matrix. 

An important property which should hold for the penalized maximum likelihood estimator is its ability to recover the graph structure of the graphical model, the so-called \emph{sparsistency} \cite{lam2009sparsistency}. We shall use this property in motivating an important step in the calculation of the bias of the likelihood in estimating the Kulback-Leibler divergence. Whether or not $\hat{\Omega}_\lambda$ is sparsistent, depends on the penalty function $p_{\lambda_{ij}}$ and typically involves how the penalty parameter $\lambda$ goes down to zero in a controlled way.


 We consider three types of penalty functions namely the lasso, the adaptive lasso and the SCAD penalty. 
 The Lasso penalty uses the $\ell_1$-norm function on the elements of the matrix $\Omega$, i.e.:
 \[ p_{\lambda} = \lambda |\omega|.\]
It was originally considered in \cite{meinshausen2006high,yuan2007model,friedman2008sparse}. It has been shown that under some assumptions on the sparsity of $\Omega$, the MPLE is sparsistent. The main disadvantage is the bias introduced by this penalty whereas the main advantage is that it is a convex function. In order to reduce the bias when the true parameter has a relative large magnitude non convex penalty functions can be used. In particular, adaptive Lasso penalty, which is proposed in \cite{zou2006adaptive}, uses the $\ell_1$-norm function  $p_{\lambda}(|\omega|)=\lambda|\omega|$ updating the tuning parameters $\lambda_{ij}$  iteratively as follows:
$$\lambda_{ij}=\lambda/|\tilde{\omega}_{ij}|^{\gamma},$$
where $\gamma>0$ is some constant and $\tilde{\omega}_{ij}$ is any consistent estimator of $\omega_{ij}$. Finally,  the SCAD penalty $p_{a,\lambda}$ is implicitly defined by means of its derivative:
 \[ p\rq{}_{a,\lambda}(\omega) = \lambda 1_{\{\omega<\lambda\}} + \frac{(a\lambda -\omega)_+ 1_{\{\omega\geq\lambda\}}}{a-1},\]
for $a >2$. The SCAD penalty was introduced by \citet{fan2001variable}. \citet{lam2009sparsistency} showed that the resulting estimator for $\Omega$ is sparsistent.  In all the penalty functions we follow the usual conventional not penalizing the diagonal. 

SCAD and  adaptive lasso penalties lead to sparsistent estimates for $\Omega$ even when the dimensionality of the problem $d$ is allowed to change with $n$, whereas the $\ell_1$ penalty requires an additional assumption of sparsity \citep{lam2009sparsistency}.

The optimazation problem (\ref{eq:mple}) for the three penalty functions described above can be solved via iterative $\ell_1$ penalized problems.  The $\ell_1$ problem is convex and fast algorithms have been proposed and
implemented in a variety of packages \citep{friedman2008sparse,huge}.


\section{Generalized Information Criterion for penalized GCGMs}
\label{sec:gic}
The Kulback-Leibler (KL) divergence can be used to measure the \lq\lq{}distance\rq\rq{} between the model and the true data generating mechanism and as such stands at the basis of several model selection methods.  It is known that a biased estimator of the KL divergence is minus twice the log-likelihood
evaluated at the point of maximum likelihood estimation.

\paragraph{Example 1}
For example, assume that the model is given by the distribution $N(0,\Omegahat_{\lambda}^{-1})$ and that the true data mechanism is given by the distribution $N(0,\Omega_0^{-1})$. Then the KL divergence from the model to the 
true distribution is (see \citep{penny2001})
\begin{equation}
\label{klloss}
\KL(\Omega_0|\Omegahat_{\lambda})=\frac12\left\{ \tr(\Omega_0^{-1}\Omegahat_{\lambda})-\log|\Omega_0^{-1}\Omegahat_{\lambda}|-d\right\}.
\end{equation}
This can be written in the scaled form as
\begin{equation}
\label{eq:dfkl}
2n\KL(\Omega_0|\Omegahat)\cong -2l(\Omegahat_{\lambda}) + n\tr\{\Omegahat(\Sigma_{0}- \mbox{S})\},
\end{equation} 
from which we see that the scaled likelihood $-2l(\hat{\Omega}_\lambda)$ is a biased estimate of the KL divergence. The scaling constant is introduced to link this quantity to the Akaike's information criterion. 
The bias is given as $n\tr(\Omegahat(\Sigma_{0}-\mbox{S}))$ and involves the unknown quantity $\Sigma_0$. 
\vspace{1cm}

For general problems, under the maximum likelihood approach, the most well-known estimate of the bias is provided by AIC, which estimates the bias as
twice the degrees of freedom which are given as the number of parameters in the model.\\
In graphical model estimation, the degrees of freedom are defined as
 \begin{equation}
 \label{eq:dfaic}
\df(\lambda)=\sum_{1\leq i<j\leq p}I(\omegahat_{ij,\lambda}\neq 0),
 \end{equation}
 which corresponds to the number of non zero elements in the estimated precision matrix. In the next section we propose to compute the degrees of freedom using the GIC. This bias term is shown in the next simple example and compared with (\ref{eq:dfaic}) and (\ref{eq:dfkl}).
 
 \paragraph{Example 2}
 Figure \ref{fig:biasterm1} shows an example in which we calculate the bias term with formula (\ref{dfgic})
for 100 values of $\lambda\in [0,0.7]$  indicated with df-GIC. This bias term is compared with the bias term in  (\ref{eq:dfaic}) indicated with 
df-AIC, and the bias term calculated using formula (\ref{eq:dfkl}) where we know the true model that generates the data. The bias terms are shown in 
dotted, broken, and straight lines for GIC, AIC and KL, respectively. When the tuning parameter increases, the bias term decreases as expected. GIC bias is most of the time below the curve of the true bias term, whereas AIC bias is always above. This is due to the
 fact that in penalized likelihood one parameter that is not estimated as zero in the precision matrix should not be counted as one. Even though the difference between bias term calculated with AIC and GIC looks small in the figure, it brings instability in the estimation of the \lq\lq{}best\rq\rq{} tuning parameter as shown in Figure \ref{fig:biasterm2}. Specifically, Figure \ref{fig:biasterm2} shows the behavior of GIC (red line) and AIC (blue line) with respect to the Kullback-Leibler divergence. AIC is minimized in a point which is not so distant from the KL but it looks very unstable, whereas GIC is really close to KL and exhibits a more stable behavior. 

\begin{figure}[!ht]
        \centering
        \begin{subfigure}[b]{0.50\textwidth}
        \includegraphics[width=\textwidth]{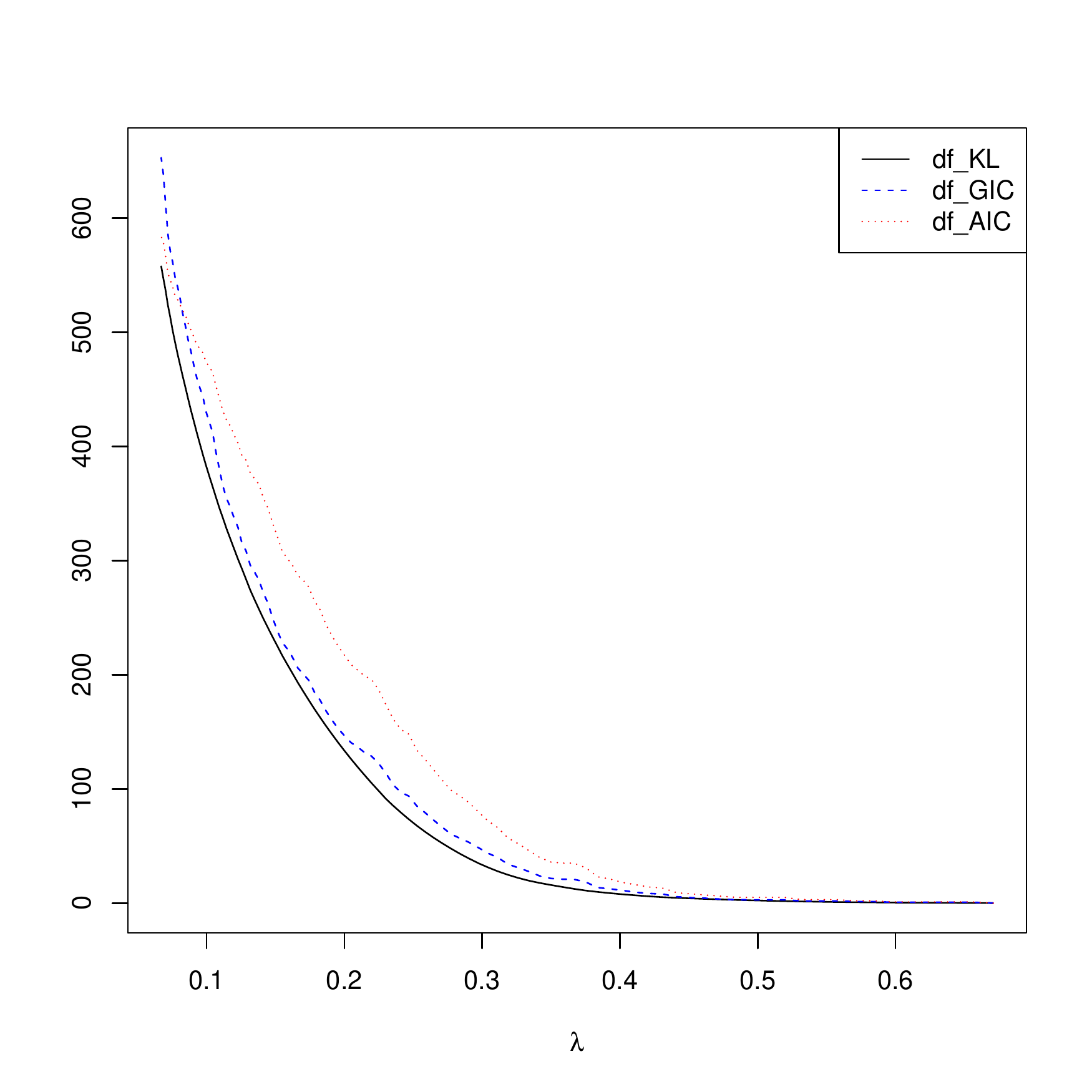}
                \caption{Bias terms}
                \label{fig:biasterm1}
        \end{subfigure}%
        ~ 
           \quad
        \begin{subfigure}[b]{0.5\textwidth}
        \includegraphics[width=\textwidth]{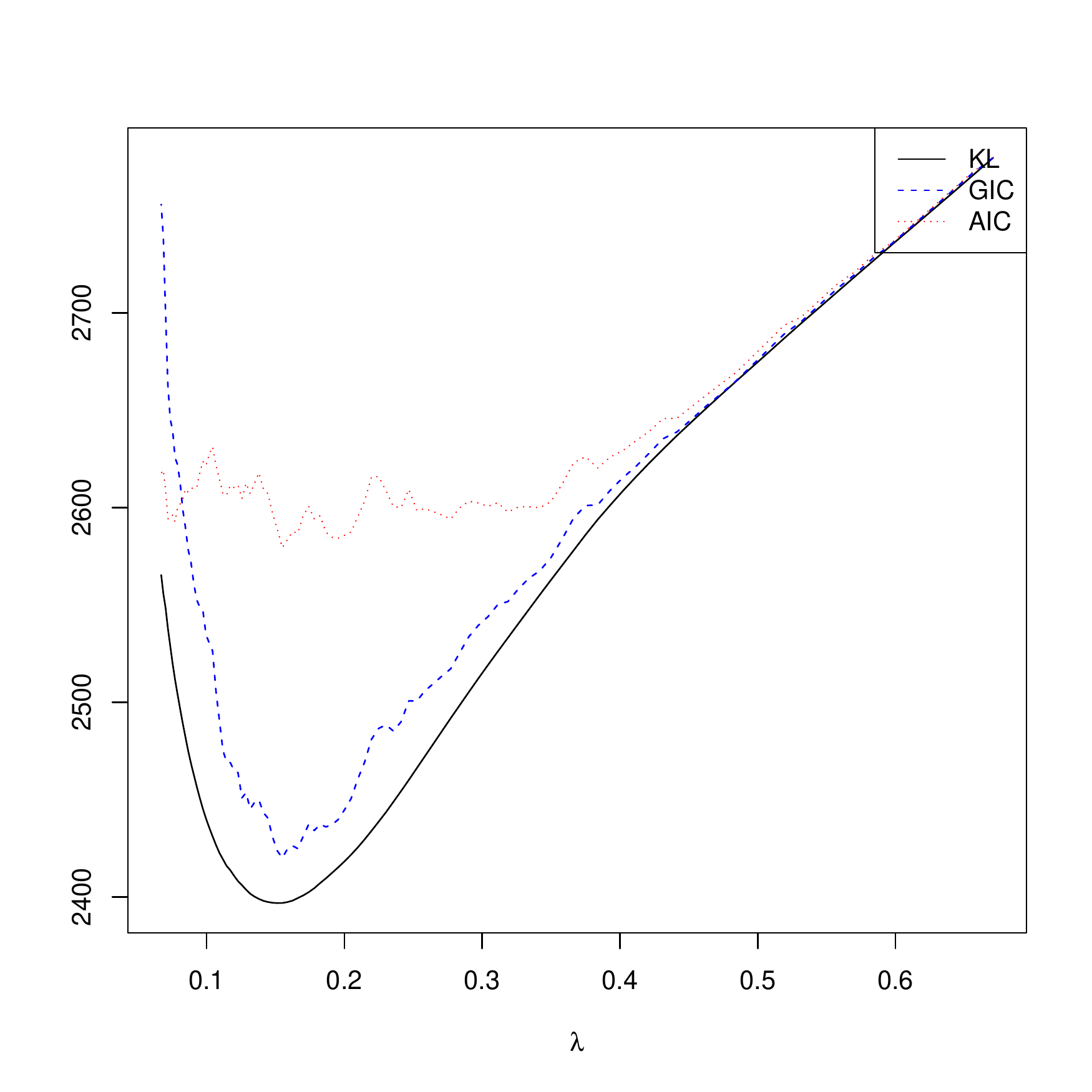}
                \caption{Information criteria}
                \label{fig:biasterm2}
        \end{subfigure}
        \caption{Bias terms (\ref{fig:biasterm1}) and information criteria (\ref{fig:biasterm2}) for AIC, GIC and KL calculated generating data from a multivariate normal distribution with n = 100 samples and d = 50 variables. \label{fig:bias_ex}}
\end{figure}

 \vspace{1cm}

The generalized information criterion (GIC), introduced in \citep{Konishi96}, provides an estimate of Kullack-Leibler divergence for wide class of estimators. Maximum likelihood estimator and penalized maximum likelihood
estimator are examples of an $M$-estimator which is defined as a solution of the system
 of equations
\begin{equation}
\label{eq:M-estimator} 
\sum_{k=1}^n \psi(q_{k},\Omega)= 0.
\end{equation}
The GIC for \textit{M}-estimator $\Omegahat$ defined by (\ref{eq:M-estimator}) is given by \citep[p. 117]{Konishi08}:
\begin{equation}
\label{eq:GICgeneral1}
\GIC=-2\sum_{k=1}^n l_{k}(\hat{\Omega}; q_k)+2\tr(\mbox{R}^{-1}\mbox{Q}),
\end{equation}
where  $\mbox{R}$ and $\mbox{Q}$ are square matrices of order $d^2$ given by 
\begin{equation}
\label{RQ}
\mbox{R}=-\frac1n\sum_{k=1}^n\{\sD\psi(q_{k},\Omegahat)\}^{\top}, \qquad \qquad
\mbox{Q}=\frac1n\sum_{k=1}^n\psi(q_{k},\Omega)\sD l_k(\Omegahat).
\end{equation}

In order to compute the bias term in (\ref{eq:GICgeneral1}) we need to derive the terms $\mbox{Q}$ and $\mbox{R}^{-1}$
for GCGMs. However, penalized estimator penalties are usually not differentiable and 
the bias term in (\ref{eq:GICgeneral1}) would require matrix multiplication for matrix of dimension $d^2$ and the storage of $r$ matrices of dimension $d^2$, where $r$ is the length of $\lambda$ term.
The former problem would make the derivation of GIC impossible while the latter would make the use of GIC infeasible for high-dimensional problems. Since in case of penalized estimator penalties are usually not differentiable, we consider another approach. Firstly, we exploit the formula obtained in case of maximum
likelihood estimator. This implies that $\lambda$ is set to be zero which means the penalty term that makes 
not differentiable the function does not enter in the calculation. Secondly, we rely on the assumption of sparsistency to derive GIC for penalized GCGMs.


%

\subsection{GIC for maximum likelihood GCGMs}
In this section we derive the formula of GIC given in (\ref{eq:GICgeneral1}) for the maximum likelihood 
GCGMs. We need to calculate the term $\tr(\mbox{R}^{-1}\mbox{Q})$ where $\mbox{R}$ and $\mbox{Q}$ are given
by (\ref{RQ}). Log-likelihood of one observation for a Gaussian copula model is up to an additive constant
 $l_k(\Omega; q_k)=\frac12\{\log|\Omega|-\tr(\Omega \tilde{\mbox{S}}_k)\}$,
where $\tilde{\mbox{S}}_k = q_k q_k^{\top}$.  Using matrix calculus \cite[pp.201,203]{Magnus07} we obtain
\begin{equation}
\label{eq:score}
\sD  l_{k}(\Omega)= \psi(q_{k},\Omega)=\frac12 \vect(\Omega^{-1} - \tilde{\mbox{S}}_{k})^{\top},
\end{equation}
and
\begin{equation}
\label{eq:term3}
\{\sD \psi(q_k;\Omega)\}^{\top}=-\frac12\Omega^{-1}\otimes \Omega^{-1}.
\end{equation}
From (\ref{eq:M-estimator})
\begin{equation}
\label{eq:m-estimator} 
\sum_{k=1}^N \psi(q_k,\Omegahat)=\frac12\sum_{k=1}^N \vect \big(\Omegahat^{-1}-\tilde{\mbox{S}}_{k}\big)^{\top}=0.
\end{equation}
It follows that from (\ref{eq:m-estimator})
\begin{equation}
\label{eq:qnew1} 
\vect\Omegahat^{-1}=\vect \tilde{\mbox{S}}.
\end{equation}

By using  (\ref{eq:score}), (\ref{eq:term3}) we obtain 
$$\mbox{Q} =-\frac{1}{4n}\sum_{k=1}^n \vect\big(\Omegahat^{-1}-\tilde{\mbox{S}}_k\big) \vect(\tilde{\mbox{S}}_k)^{\top},$$
and by applying (\ref{eq:qnew1}) we have 
$$\mbox{Q}= \frac{1}{4n}\sum_{k=1}^n \vect \tilde{\mbox{S}}_k \vect(\tilde{\mbox{S}}_k)^{\top} - \frac{1}{4} \vect \tilde{\mbox{S}} \vect(\tilde{\mbox{S}})^{\top}.$$

We obtain expression for $\mbox{R}^{-1}$ from (\ref{eq:term3})
\begin{equation}
\label{Rinv}
\mbox{R}^{-1}=2\Omega\otimes \Omega.
\end{equation}

Now, it follows that 
\begin{equation}
\label{eqgic:mle1}
\tr(\mbox{R}^{-1}\mbox{Q})=\frac{1}{2n}\sum_{k=1}^n\tr\{(\Omegahat\otimes \Omegahat) \vect \tilde{\mbox{S}}_{k}\vect \tilde{\mbox{S}}_{k}^{\top}\}- \frac{1}{2}\tr\{(\Omegahat\otimes \Omegahat) \vect \tilde{\mbox{S}}\vect \tilde{\mbox{S}}^{\top}\},
\end{equation}

which gives the bias term in case of maximum likelihood for GCGMs.
 Here, $\otimes$ is the Kronecker product of matrices and $\vect$ is the vectorization operator which transforms a matrix into a column vector obtained by stacking the columns of the matrix on top of one another.

\subsection{GIC for penalized maximum likelihood GCGMs}

We, firstly, write (\ref{eqgic:mle1}) in a different form to obtain GIC for penalized maximum likelihood GCGMs. By using the identity $\tr(\mbox{A}\mbox{B})=\tr(\mbox{B}\mbox{A})$ and noting that for vector $x$ and matrix $\mbox{A}$ holds $\tr(\mbox{A}xx^{\top})=x^{\top}\mbox{A}x$ we obtain that formula (\ref{eqgic:mle1}) can be written as 
\begin{equation}
\nonumber
\tr(\mbox{R}^{-1}\mbox{Q})=\frac{1}{2n}\sum_{k=1}^n \vect \tilde{\mbox{S}}_{k}^{\top}(\Omegahat\otimes \Omegahat)
\vect \tilde{\mbox{S}}_{k}- \frac{1}{2}\vect \tilde{\mbox{S}}^{\top}(\Omegahat\otimes \Omegahat)\vect \tilde{\mbox{S}}.
\end{equation}
We can write the formula as
\begin{equation}
\label{gic:mle2}
\tr(\mbox{R}^{-1}\mbox{Q})=\frac{1}{2n}\sum_{k=1}^n T(\tilde{\mbox{S}}_k)- \frac{1}{2}T(\tilde{\mbox{S}}),
\end{equation}
where 
$$T(\mbox{A})=\vect \mbox{A}^{\top}(\Omegahat\otimes \Omegahat) \vect \mbox{A},$$
and $\mbox{A}$ is a $d\times d$ matrix. 

\begin{lemma}
\label{lemma1}
Let $\mathrm{A}$ and $\Omega$ be a symmetric matrices of order $d$. Then the following identity holds
$$(\Omega\otimes\Omega)\vect(\mathrm{A})=\mbox{M}_{d}(\Omega\otimes\Omega)\vect(\mathrm{A}),$$
where $\mathrm{M}_{d}= \frac{1}{2}(\mathrm{I}_{d^2}+\mathrm{K}_{d})$, and $\mathrm{I}_{d^2}$ and $\mathrm{K}_{d}$ are identity matrix and commutation matrix of order $d^2$, respectively.
\end{lemma}
\begin{proof}
Commutation matrix $\mbox{K}_{d}$ is defined as a matrix that has the property $\mbox{K}_{d} \vect(\mbox{A}) = \vect(\mbox{A})^{\top}$. By substituting $\mbox{M}_{d} = \frac{1}{2}(\mbox{I}_{d^2}+\mbox{K}_{d})$ in the equality we obtain that it is equivalent to 
$$(\Omega\otimes\Omega)\vect(\mbox{A})=\mbox{M}_{d}(\Omega\otimes\Omega)\vect(\mbox{A}).$$
To show this, we use identities  $\vect(\mbox{A}\mbox{B}\mbox{C})=(\mbox{C}^{\top}\otimes \mbox{A})\vect \mbox{B}$, $\mbox{K}_{d} \vect(\mbox{A}) = \vect(\mbox{A})^{\top}$ and symmetry of $\mbox{A}$ and $\Omega$ 
$$\mbox{M}_{d}(\Omega\otimes \Omega) \vect \mbox{A} = \mbox{K}_{d} \vect(\Omega \mbox{A}\Omega) = \vect(\Omega \mbox{A}\Omega)^{\top} = \vect(\Omega \mbox{A}\Omega) = 
(\Omega\otimes \Omega) \vect \mbox{A}.$$
\end{proof}

According to Lemma 1 it follows that 
$$T(\mbox{A})=\vect \mbox{A}^{\top} \mbox{M}_{d}(\Omegahat\otimes \Omegahat)\vect \mbox{A}.$$

Since $\mbox{M}_d(\Omegahat\otimes\Omegahat)$ is an estimate of the asymptotic covariance matrix of $\Omegahat$ \citep{fried2009robust}, we use an asymptotic argument to propose an estimate in case of
MPLE. To obtain the formula for the penalized estimator we assume standard conditions like in \citep{lam2009sparsistency} that guarantees sparsistent shrinkage estimator. These conditions 
imply that $\lambda\rightarrow 0$ when $n\rightarrow \infty$, so we can use formula (\ref{gic:mle2}), derived for the maximum likelihood case, as an approximation in the penalized case. By sparsistency with probability one
 the zero coefficients will be estimated as zero when $n$ tends to infinity. This means that asymptotically the covariances between zero elements and nonzero elements are equal to zero. Thus, to obtain the term $T(\mbox{A})$ for
 the shrinkage estimator we do not only consider the expression $\Omegahat_{\lambda}$ in formula (\ref{gic:mle2}), but we also set the elements of the matrix $\mbox{M}_{d}(\Omegahat_{\lambda}\otimes\Omegahat_{\lambda})$ 
 corresponding to covariances between zero and nonzero elements to zero.

 \begin{lemma}
\label{lemma2}
Let $\mathrm{A}$ be a symmetric matrix of order $d$ and $x,y$ any vectors of dimension $d$. Then setting $i$-th row (column) of matrix $\mathrm{A}$ to zero in the bilinear form 
$$x^{\top}\mathrm{A}y,$$
is equivalent to setting $i$-th entry of vector $x$ ($y$) to zero.
\end{lemma}
\begin{proof}
The result is obtained by straightforward calculation.
\end{proof}

According to Lemma \ref{lemma2} this is equivalent to setting the corresponding entries of vectors 
$\vect A$ to zero,
\begin{equation}
\label{T} 
T_{\lambda}(\mbox{A})=\vect (\mbox{A}\circ \mathrm{I}_{\lambda})^{\top} \mbox{M}_{d}(\Omegahat_{\lambda}\otimes\Omegahat_{\lambda})\vect (\mbox{A}\circ \mathrm{I}_{\lambda})=\vect (\mbox{A}\circ \mbox{I}_{\lambda})^{\top} (\Omegahat_{\lambda}\otimes\Omegahat_{\lambda})\vect (\mbox{A}\circ \mbox{I}_{\lambda}),
\end{equation}
where $\circ$ is the Schur or Hadamard product of matrices and $\mathrm{I}_{\lambda}$ is the indicator matrix, whose entry is $1$ if the corresponding entry in the precision matrix $\Omegahat_{\lambda}$ is nonzero and zero if the corresponding entry in the precision matrix is zero.
The second equality in (\ref{T}) follows from Lemma 1. By substituting $T_{\lambda}(\tilde{\mbox{S}})$ and $T_{\lambda}(\tilde{\mbox{S}}_k)$ instead of  
$T(\tilde{\mbox{S}})$ and $T(\tilde{\mbox{S}}_k)$ in (\ref{gic:mle2}) we obtain:
\begin{equation}
\label{gic:mple}
\tr(\mbox{R}^{-1}\mbox{Q})=\frac{1}{2n}\sum_{k=1}^n T_{\lambda}(\mbox{S}_k)- \frac{1}{2}T_{\lambda}(\mbox{S}).
\end{equation}
Implementation of  this formula is computationally too expensive so we rewrite it in a different way.



\subsection{Efficient calculation of degrees of freedom for GCGMs}
\label{sec:egic}

By using the identity $\vect(\mbox{ABC})=(\mbox{C}^{\top}\otimes \mbox{A})\vect \mbox{B}$ we obtain 
\begin{equation}
\label{efficient}
T_{\lambda}(\mbox{A})= \vect (\mbox{A}\circ \mbox{I}_{\lambda})^{\top} (\Omegahat_{\lambda}\otimes\Omegahat_{\lambda})\vect (\mbox{A}\circ \mbox{I}_{\lambda})= \vect (\mbox{A}\circ \mbox{I}_{\lambda})^{\top}\vect\{\Omegahat_{\lambda}(\mbox{A}\circ \mbox{I}_{\lambda})\Omegahat_{\lambda}\}
\end{equation}

From (\ref{efficient}) we propose the following estimator of the Kullback-Leibler divergence of the penalized Gaussian copula model to the true distribution  
\begin{equation}
\label{eq:GICgeneral}
\GIC(\lambda) = -2l(\Omegahat_{\lambda})+2\widehat{\df}_{\GIC},
\end{equation}
where
\begin{equation}
\label{dfgic}
 \widehat{\df}_{\GIC} =  \frac{1}{2n}\sum_{k=1}^n \vect(\tilde{\mbox{S}}_k\circ \mathrm{I}_{\lambda})^{\top}\vect\{\Omegahat_{\lambda}(\tilde{\mbox{S}}_k\circ \mathrm{I}_{\lambda})\Omegahat_{\lambda}\}
- \frac{1}{2}\vect(\tilde{\mbox{S}}\circ \mathrm{I}_{\lambda})^{\top}\vect\{\Omegahat_{\lambda}(\tilde{\mbox{S}}\circ \mbox{I}_{\lambda})\Omegahat_{\lambda}\},
\end{equation}

In the case of maximum likelihood estimator, i.e. when $\lambda=0$, with probability one every element of matrix $\mbox{I}_{\lambda}$ is equal to one so we obtain the same formula as in \ref{gic:mle2}. The calculation of GIC in (\ref{eq:GICgeneral}) is feasible for high-dimensional data. The penalized estimator is derived under the assumption of sparsistency of the estimator, while for maximum likelihood estimator we do not need any assumption. To estimate the tuning parameter $\lambda$, we consider a set of $r$ positive values and we choose the value that minimizes (\ref{eq:GICgeneral}).  
The bias term estimator in (\ref{dfgic}) can be thought as a measure of model complexity and used as a degrees of freedom in other information criteria. Although its use in this way is not formally justified, as it is for the estimate of the Kulback-Leibler divergence, simulations in Section \ref{sec:simulations} show that this can be useful to improve the performance of the KL-estimator. For example,
we define GBIC as follows:

\begin{equation}
\label{gbic}
\GIC(\lambda) = -2l(\Omegahat_{\lambda})+\mbox{log}(n)\widehat{\df}_{\GIC}.
\end{equation}

\section{Numerical results}
\label{sec:sim}

\subsection{Simulation study}
\label{sec:simulations}

We perform a simulation study to evaluate the performance of the proposed estimator in comparison with other estimators proposed in the literature. In particular, we evaluate the performance of 
the proposed estimator which we called GIC (\ref{eq:GICgeneral}), two fold cross validation (CV), the Akaike's information criterion (AIC), the Schwarz's information criterion (BIC),  the Schwarz's information criterion adjusted with 
degrees of freedom in formula (\ref{dfgic}) (GBIC (\ref{gbic})). Moreover, we compute the performance of the Kullback-Leibler oracle which is calculated from formula (\ref{eq:dfkl}). These estimators are evaluated with respect to: 
\begin{itemize}
\item[i)] estimation quality;
\item[ii)] support recovery of the precision matrix.
\end{itemize}

The estimation quality is measured by the Kullback-Leibler loss function:
 $$\mbox{KL}(\Omega, \hat{\Omega}) = 
 \mbox{tr}(\Omega^{-1} \hat{\Omega}) - \mbox{log}|\Omega^{-1} \hat{\Omega}| - p,$$
where $\Omega^{-1} =\Sigma$ is the true variance covariance matrix.
Moreover, the estimation quality is measured by the following matrix norms: the operator norm, the matrix $\ell_1$ norm, the Frobenius norm. Obviously, the smaller is the loss, the better the quality of the estimator.

The support recovery of the precision matrix is evaluated by the following scores:

$$\mbox{Specificity} = \frac{\mbox{TN}}{\mbox{TN}+\mbox{FP}}, \qquad \qquad \mbox{\mbox{Sensitivity}} = \frac{\mbox{TP}}{\mbox{TP}+\mbox{FN}},$$
and
$$\mbox{MCC} = \frac{(\mbox{TP}\times \mbox{TN}) - (\mbox{FN}\times \mbox{FP})}{\sqrt{(\mbox{TP}\times \mbox{TN})(\mbox{FN}\times \mbox{FP})}},$$
where 
$\mbox{TP}$ are the true positives, $\mbox{FP}$ are the false positives, $\mbox{TN}$ are the true negatives and 
$\mbox{FN}$ are the false negatives. The larger the score value, the better the classification performance.

The cross validation is implemented as follows. We divided the sample in two parts with the same numbers of elements, namely the training sample and the validation sample. We used the training data to compute a series of estimators with 200 different values of $\lambda$ and selected the one with the smallest likelihood loss on the validation sample, where the likelihood loss is defined by
$$L(\Sigma, \hat{\Omega}_\lambda) = \mbox{tr}(\hat{\Omega}_\lambda \Sigma) - \mbox{log}|\hat{\Omega}_\lambda|.$$

We consider three models as follows:
\begin{itemize}
\item Model 1. A banded graph with bandwidth equal to 4;
\item Model 2. A random graph where each edge is present with probability $3/d$, where $d$ is the number of variables;
\item Model 3. A random graph where each edge is present with probability $0.9$.
\end{itemize}

Model 1 has a banded structure where the values of the entries between $node_i$ and $node_{i + 5}$ are zeroes.
Model 2 is an example of sparse matrix without any special pattern. Model 3 serves as a dense matrix example. For each model, we generate a sample of size $n = 100$ from a multivariate normal distribution.
We compute the glasso, adaptive lasso and scad estimators of these data. We consider different values of 
$d = 30, 60, 90, 120$ and 100 replicates. The simulation is conducted for three penalized Gaussian Graphical models: glasso, adaptive lasso and scad. Results are shown in the appendix (\ref{app3}).

These tables show how the norm changes when we select the tuning parameters by minimizing one of the following measure: GIC, CV, AIC, BIC, GBIC, KL-oracle. The latter measure is computed by formula (\ref{eq:dfkl}). This can be refereed to be the true bias term correction
for Gaussian graphical model. It is necessary to know the true structure of the graph to compute this 
estimate. Since our estimator is based on the minimization of KL, we expect similar performances between KL-oracle and GIC. This can be seen in all the Tables we report. This should also be the case for AIC if the bias correction term is adequate. Our simulations show that AIC and CV are not adequate when $p$ tends to $n$. Similar simulation studies have been conducted in \cite{cai2011constrained}, \cite{fan2009network}, \cite{rothman2008sparse}. The main difference is that in this paper we are concerned about selection of the tuning parameter whereas in the just cited papers there are proposal of new estimators of the precision matrix. The results can be summarized as follows: i) GIC performs overall better in terms of KL-loss function;
ii) there is no clear pattern with the other norm loss; iii)  the degrees of freedom as proposed in formula (\ref{dfgic}) improves the performance of BIC, i.e. GBIC performance is, generally, better than BIC. 
Our method performs better for dense graphs in terms of specificity when used in combination with adaptive lasso. Note that for dense graphs we report only specificity since other two measures are not useful due to the structure of the graph. In this case we have indicated Sensitivity and MCC with N/A.

\section{Discussion}
High-dimensional data arises in many fields of science. The study of the structure, in terms of conditional independence graph, can be carried out with penalized Gaussian copula graphical models. 
These models allow us to visualize with a graph the conditional independence among the set of random variables. 
These models make use of the sparsity assumption, i.e. many parameters are not statistically significant. In this setting, we proposed an estimator of the  Kullback-Leibler loss function which can be used for the main penalized likelihood approaches namely glasso, adaptive lasso and scad in order to model selection in GCGMs. We select the tuning parameter that minimizes the GIC. In a simulation study, we showed the performance of the proposed estimator of the bias term can be used with other information criteria  (such as BIC) to improve their performance in terms of graph selection.

\appendix

     \section{Simulation study results}
\label{app3}

\renewcommand{\arraystretch}{1.2}
\begin{landscape}
\begin{table}[!ht]
\centering
\begin{tabular}{r@{\hskip 0.5in} r@{ (}r@{) \hskip 0.3in}r@{ (}r@{) \hskip 0.3in}r@{ (}r@{) \hskip 0.3in}r@{ (}r@{) \hskip 0.3in}r@{ (}r@{) \hskip 0.3in}r@{ (}r@{) \hskip 0.3in}}
  &\multicolumn{2}{c}{\itshape gic} &  \multicolumn{2}{c}{\itshape cv}& \multicolumn{2}{c}{\itshape aic}& \multicolumn{2}{c}{\itshape bic} & \multicolumn{2}{c}{\itshape gbic} & \multicolumn{2}{c}{\itshape kl oracle}\\ 
  \hline
&\multicolumn{12}{c}{}\\
p&\multicolumn{12}{c}{\textbf{ Kullback Leibler loss function}}\\
&\multicolumn{12}{c}{}\\
  30 & 1.82 & 0.13 & 1.72 & 0.12 & $\bold{1.60}$ & 0.10 & 3.19 & 0.77 & 2.94 & 0.39 & 1.58 & 0.09 \\ 
  60 & $\bold{4.14}$ & 0.17 & 4.21 & 0.25 & 4.40 & 0.26 & 7.71 & 0.35 & 6.49 & 0.53 & 3.89 & 0.16 \\ 
  90 & $\bold{6.78}$ & 0.24 & 6.98 & 0.33 & 8.20 & 0.70 & 11.72 & 0.40 & 10.43 & 0.60 & 6.51 & 0.22 \\ 
  120 & $\bold{9.42}$ & 0.26 & 9.84 & 0.37 & 11.37 & 1.96 & 15.89 & 0.38 & 14.30 & 0.68 & 9.21 & 0.23 \\ 
  &\multicolumn{12}{c}{}\\
  &\multicolumn{12}{c}{\textbf{Operator norm}}\\
  &\multicolumn{12}{c}{}\\
  30  & 3.87 & 0.11 & 3.75 & 0.14 & $\bold{3.30}$ & 0.12 & 4.41 & 0.29 & 4.38 & 0.12 & 3.40 & 0.12 \\ 
  60  & 3.80 & 0.06 & 3.82 & 0.09 & $\bold{3.16}$ & 0.11 & 4.43 & 0.04 & 4.27 & 0.07 & 3.56 & 0.06 \\ 
  90  & 3.82 & 0.05 & 3.87 & 0.07 & $\bold{3.17}$ & 0.16 & 4.38 & 0.03 & 4.28 & 0.05 & 3.65 & 0.05 \\ 
  120  & 3.83 & 0.05 & 3.91 & 0.06 & $\bold{3.45}$ & 0.38 & 4.37 & 0.02 & 4.28 & 0.04 & 3.71 & 0.03 \\ 
  &\multicolumn{12}{c}{}\\
  &\multicolumn{12}{c}{\textbf{Matrix $\ell_1$ norm}}\\
  &\multicolumn{12}{c}{}\\
   30 & 4.65 & 0.16 & 4.64 & 0.17 & $\bold{4.60}$ & 0.23 & 4.82 & 0.13 & 4.81 & 0.08 & 4.59 & 0.22 \\ 
   60 & 4.66 & 0.11 & 4.64 & 0.11 & 5.26 & 0.23 & 4.53 & 0.02 & $\bold{4.51}$ & 0.04 & 4.82 & 0.15 \\ 
   90 & 4.79 & 0.11 & 4.73 & 0.11 & 5.97 & 0.38 & $\bold{4.46}$ & 0.02 & $\bold{4.46}$ & 0.03 & 5.01 & 0.12 \\ 
   120 & 4.90 & 0.13 & 4.77 & 0.12 & 5.82 & 0.93 & $\bold{4.44}$ & 0.02 & $\bold{4.44}$ & 0.03 & 5.10 & 0.13 \\ 
   &\multicolumn{12}{c}{}\\
   &\multicolumn{12}{c}{\textbf{Frobenius norm}}\\
   &\multicolumn{12}{c}{}\\
  30 & 6.65 & 0.19 & 6.42 & 0.25 & $\bold{5.66}$ & 0.15 & 7.84 & 0.60 & 7.75 & 0.25 & 5.80 & 0.17 \\ 
  60 & 9.42 & 0.15 & 9.50 & 0.24 & $\bold{8.02}$ & 0.14 & 11.44 & 0.14 & 10.94 & 0.22 & 8.78 & 0.10 \\ 
  90 & 11.77 & 0.17 & 11.96 & 0.25 & $\bold{10.17}$ & 0.31 & 14.03 & 0.13 & 13.61 & 0.20 & 11.18 & 0.11 \\ 
  120 & 13.68 & 0.18 & 14.06 & 0.23 & $\bold{12.80}$ & 0.92 & 16.28 & 0.10 & 15.83 & 0.19 & 13.21 & 0.09 \\ 
   &\multicolumn{12}{c}{}\\
   \hline
\end{tabular}
\caption{Comparison of average (SE) matrix losses for MODEL 1 which is a BAND GRAPH
where GLASSO has been used to estimate the precision matrix 
over 100 replications \label{tab:bandgraph_norm} }
\end{table}
\end{landscape}

\renewcommand{\arraystretch}{1.2}
\begin{landscape}
\begin{table}[!ht]
\centering
\begin{tabular}{r@{\hskip 0.5in} r@{ (}r@{) \hskip 0.3in}r@{ (}r@{) \hskip 0.3in}r@{ (}r@{) \hskip 0.3in}r@{ (}r@{) \hskip 0.3in}r@{ (}r@{) \hskip 0.3in}r@{ (}r@{) \hskip 0.3in}}
  &\multicolumn{2}{c}{\itshape gic} &  \multicolumn{2}{c}{\itshape cv}& \multicolumn{2}{c}{\itshape aic}& \multicolumn{2}{c}{\itshape bic} & \multicolumn{2}{c}{\itshape gbic} & \multicolumn{2}{c}{\itshape kl oracle}\\ 
  \hline
&\multicolumn{12}{c}{}\\
p&\multicolumn{12}{c}{\textbf{Specificity}}\\
&\multicolumn{12}{c}{}\\
30 & 43.41 & 2.88 & 40.99 & 2.73 & 36.45 & 1.51 & $\bold{68.81}$ & 17.30 & 62.45 & 8.01 & 37.20 & 1.48 \\ 
  60 & 31.03 & 2.05 & 31.65 & 2.52 & 21.71 & 0.85 & $\bold{83.55}$ & 8.34 & 63.46 & 8.42 & 25.73 & 0.84 \\ 
  90 & 25.53 & 1.61 & 27.31 & 2.85 & 16.01 & 1.64 & $\bold{82.72}$ & 6.93 & 65.88 & 7.83 & 21.06 & 0.61 \\ 
  120 & 21.67 & 1.40 & 25.05 & 2.70 & 17.28 & 5.63 & $\bold{83.80}$ & 5.55 & 66.73 & 8.06 & 18.46 & 0.47 \\ 
  &\multicolumn{12}{c}{}\\
  &\multicolumn{12}{c}{\textbf{Sensitivity}}\\
  &\multicolumn{12}{c}{}\\
  30  & 67.25 & 3.46 & 70.80 & 4.44 & $\bold{80.23}$ & 3.01 & 30.75 & 21.43 & 40.46 & 11.45 & 78.84 & 3.51 \\ 
  60  & 63.10 & 2.76 & 62.32 & 3.42 & $\bold{77.94}$ & 2.56 & 11.45 & 6.30 & 32.20 & 8.75 & 70.06 & 2.51 \\ 
  90  & 60.77 & 2.29 & 58.72 & 3.33 & $\bold{75.95}$ & 3.78 & 10.09 & 4.90 & 25.37 & 7.14 & 66.43 & 1.84 \\ 
  120  & 59.65 & 1.93 & 56.09 & 2.83 & $\bold{67.51}$ & 8.62 & 8.89 & 3.47 & 22.99 & 6.24 & 63.80 & 1.66 \\ 
  &\multicolumn{12}{c}{}\\
  &\multicolumn{12}{c}{\textbf{MCC}}\\
  &\multicolumn{12}{c}{}\\
  30 & 33.34 & 3.85 & 31.66 & 3.56 & 28.72 & 2.99 & 30.76 & 9.24 & $\bold{36.77}$ & 5.65 & 29.43 & 3.09 \\ 
   60 & 32.08 & 2.29 & 32.42 & 2.37 & 24.21 & 1.66 & 26.97 & 7.07 & $\bold{38.91}$ & 3.95 & 28.01 & 1.53 \\ 
   90 & 30.42 & 1.67 & 31.47 & 2.17 & 21.50 & 1.93 & 26.22 & 5.77 & $\bold{36.98}$ & 4.03 & 26.91 & 1.06 \\ 
   120 & 28.57 & 1.38 & 30.81 & 1.96 & 23.95 & 5.23 & 25.60 & 4.69 & $\bold{36.20}$ & 3.55 & 25.99 & 0.87 \\ 
   &\multicolumn{12}{c}{}\\
   \hline
\end{tabular}
\caption{Comparison of average (SE) support recovery for MODEL 1 which is a BAND GRAPH
where GLASSO has been used to estimate the precision matrix 
over 100 replications \label{tab:bandgraph} }
\end{table}
\end{landscape}

\renewcommand{\arraystretch}{1.2}
\begin{landscape}
\begin{table}[!ht]
\centering
\begin{tabular}{r@{\hskip 0.5in} r@{ (}r@{) \hskip 0.3in}r@{ (}r@{) \hskip 0.3in}r@{ (}r@{) \hskip 0.3in}r@{ (}r@{) \hskip 0.3in}r@{ (}r@{) \hskip 0.3in}r@{ (}r@{) \hskip 0.3in}}
  &\multicolumn{2}{c}{\itshape gic} &  \multicolumn{2}{c}{\itshape cv}& \multicolumn{2}{c}{\itshape aic}& \multicolumn{2}{c}{\itshape bic} & \multicolumn{2}{c}{\itshape gbic} & \multicolumn{2}{c}{\itshape kl oracle}\\ 
  \hline
&\multicolumn{12}{c}{}\\
p&\multicolumn{12}{c}{\textbf{ Kullback Leibler loss function }}\\
&\multicolumn{12}{c}{}\\
30  & 1.88 & 0.13 & 2.66 & 0.26 & $\bold{1.86}$ & 0.13 & 4.32 & 0.10 & 2.51 & 0.27 & 1.85 & 0.13 \\ 
  60 & $\bold{4.08}$ & 0.17 & 6.32 & 0.33 & 6.37 & 2.01 & 8.50 & 0.15 & 5.35 & 0.46 & 4.03 & 0.16 \\ 
  90 & $\bold{6.51}$ & 0.22 & 10.17 & 0.27 & 12.23 & 0.15 & 12.85 & 0.20 & 8.25 & 0.55 & 6.50 & 0.22 \\ 
   120 &$\bold{9.20}$ & 0.28 & 13.93 & 0.32 & 16.38 & 0.19 & 17.21 & 0.22 & 11.42 & 0.65 & 9.17 & 0.28 \\ 
  &\multicolumn{12}{c}{}\\
  &\multicolumn{12}{c}{\textbf{Operator norm}}\\
  &\multicolumn{12}{c}{}\\
  30  & 3.70 & 0.10 & 4.18 & 0.13 & $\bold{3.67}$ & 0.11 & 4.73 & 0.02 & 4.10 & 0.14 & 3.66 & 0.11 \\ 
  60  & $\bold{3.57}$ & 0.08 & 4.20 & 0.07 & 4.07 & 0.48 & 4.51 & 0.01 & 4.00 & 0.11 & 3.53 & 0.08 \\ 
  90  & $\bold{3.55}$ & 0.07 & 4.23 & 0.03 & 4.42 & 0.01 & 4.46 & 0.01 & 3.98 & 0.08 & 3.54 & 0.07 \\ 
  120  & $\bold{3.55}$ & 0.06 & 4.24 & 0.03 & 4.40 & 0.01 & 4.44 & 0.01 & 4.01 & 0.07 & 3.57 & 0.06 \\ 
  &\multicolumn{12}{c}{}\\
  &\multicolumn{12}{c}{\textbf{Matrix $\ell_1$ norm}}\\
  &\multicolumn{12}{c}{}\\
   30 & 4.53 & 0.18 & 4.70 & 0.11 & $\bold{4.51}$ & 0.19 & 4.97 & 0.02 & 4.67 & 0.13 & 4.51 & 0.19 \\ 
   60 & 4.50 & 0.17 & $\bold{4.42}$ & 0.04 & 4.53 & 0.12 & 4.56 & 0.01 & $\bold{4.42}$ & 0.09 & 4.51 & 0.17 \\ 
   90 & 4.65 & 0.14 & $\bold{4.36}$ & 0.04 & 4.44 & 0.01 & 4.48 & 0.01 & 4.42 & 0.07 & 4.66 & 0.14 \\ 
   120 & 4.84 & 0.17 & $\bold{4.32}$ & 0.03 & 4.41 & 0.01 & 4.45 & 0.01 & 4.41 & 0.07 & 4.80 & 0.15 \\ 
   &\multicolumn{12}{c}{}\\
   &\multicolumn{12}{c}{\textbf{Frobenius norm}}\\
   &\multicolumn{12}{c}{}\\
  30 & 6.44 & 0.13 & 7.34 & 0.22 & $\bold{6.38}$ & 0.14 & 8.55 & 0.05 & 7.20 & 0.23 & 6.38 & 0.14 \\ 
  60 & $\bold{8.91}$ & 0.16 & 10.65 & 0.17 & 10.41 & 1.39 & 11.74 & 0.05 & 10.08 & 0.28 & 8.81 & 0.15 \\ 
   90 & $\bold{10.96}$ & 0.16 & 13.30 & 0.12 & 14.18 & 0.05 & 14.37 & 0.06 & 12.40 & 0.27 & 10.92 & 0.17 \\ 
  120 & $\bold{12.70}$ & 0.17 & 15.49 & 0.12 & 16.40 & 0.05 & 16.62 & 0.06 & 14.48 & 0.27 & 12.78 & 0.14 \\ 
   &\multicolumn{12}{c}{}\\
   \hline
\end{tabular}
\caption{Comparison of average (SE) matrix losses for MODEL 1 which is a BAND GRAPH 
where ADAPTIVE LASSO has been used to estimate the precision matrix 
over 100 replications \label{tab:bandgraph_adaptive} }
\end{table}
\end{landscape}

\renewcommand{\arraystretch}{1.2}
\begin{landscape}
\begin{table}[!ht]
\centering
\begin{tabular}{r@{\hskip 0.5in} r@{ (}r@{) \hskip 0.3in}r@{ (}r@{) \hskip 0.3in}r@{ (}r@{) \hskip 0.3in}r@{ (}r@{) \hskip 0.3in}r@{ (}r@{) \hskip 0.3in}r@{ (}r@{) \hskip 0.3in}}
  &\multicolumn{2}{c}{\itshape gic} &  \multicolumn{2}{c}{\itshape cv}& \multicolumn{2}{c}{\itshape aic}& \multicolumn{2}{c}{\itshape bic} & \multicolumn{2}{c}{\itshape gbic} & \multicolumn{2}{c}{\itshape kl oracle}\\ 
  \hline
&\multicolumn{12}{c}{}\\
p&\multicolumn{12}{c}{\textbf{Specificity}}\\
&\multicolumn{12}{c}{}\\
30 & 59.25 & 3.70 & $\bold{73.44}$ & 6.82 & 58.09 & 3.78 & 0.00 & 0.00 & 71.37 & 6.64 & 58.04 & 3.87 \\ 
  60 & 44.31 & 2.76 & $\bold{80.90}$ & 7.01 & 17.44 & 21.10 & 0.00 & 0.00 & 68.37 & 7.21 & 42.69 & 2.76 \\ 
  90 & 36.94 & 2.44 &$\bold{86.25}$ & 6.39 & 0.00 & 0.00 & 0.00 & 0.00 & 65.13 & 6.86 & 36.24 & 2.39 \\ 
  120 & 31.12 & 2.30 & $\bold{90.48}$ & 6.73 & 0.00 & 0.00 & 0.00 & 0.00 & 64.05 & 6.55 & 31.89 & 1.57 \\ 
  &\multicolumn{12}{c}{}\\
  &\multicolumn{12}{c}{\textbf{Sensitivity}}\\
  &\multicolumn{12}{c}{}\\
  30  & 57.97 & 4.31 & 33.71 & 7.72 &$\bold{59.84}$ & 4.11 & 0.00 & 0.00 & 36.05 & 8.09 & 59.99 & 4.20 \\ 
   60  & $\bold{55.51}$ & 3.28 & 16.09 & 5.73 & 23.61 & 28.54 & 0.00 & 0.00 & 30.30 & 7.35 & 57.61 & 3.29 \\ 
   90  & $\bold{54.33}$ & 2.89 & 8.73 & 3.53 & 0.00 & 0.00 & 0.00 & 0.00 & 29.34 & 5.89 & 55.07 & 2.92 \\ 
   120  & $\bold{53.49}$ & 2.61 & 5.93 & 2.83 & 0.00 & 0.00 & 0.00 & 0.00 & 26.48 & 5.30 & 52.55 & 2.37 \\ 
   &\multicolumn{12}{c}{}\\
  &\multicolumn{12}{c}{\textbf{MCC}}\\
  &\multicolumn{12}{c}{}\\
   30 & $\bold{44.71}$ & 3.54 & 39.66 & 5.14 & 44.69 & 3.61 & 0.00 & 0.00 & 40.01 & 5.31 & 44.71 & 3.68 \\ 
   60 & $\bold{41.00}$ & 2.23 & 32.08 & 5.86 & 16.62 & 20.09 & 0.00 & 0.00 & 40.01 & 3.91 & 40.65 & 2.26 \\ 
   90 & 38.29 & 1.99 & 25.29 & 5.24 & 0.00 & 0.00 & 0.00 & 0.00 & $\bold{39.93}$ & 3.30 & 38.05 & 1.92 \\ 
   120 & 35.27 & 1.59 & 21.50 & 5.14 & 0.00 & 0.00 & 0.00 & 0.00 & $\bold{38.30}$ & 2.94 & 35.57 & 1.39 \\ 
   &\multicolumn{12}{c}{}\\
   \hline
\end{tabular}
\caption{Comparison of average (SE) support recovery for MODEL 1 which is a BAND GRAPH
where ADAPTIVE LASSO has been used to estimate the precision matrix 
over 100 replications \label{tab:bandgraph_adaptive_roc} }
\end{table}
\end{landscape}

\renewcommand{\arraystretch}{1.2}
 \begin{landscape}
 \begin{table}[!ht]
 \centering
 \begin{tabular}{r@{\hskip 0.5in} r@{ (}r@{) \hskip 0.3in}r@{ (}r@{) \hskip 0.3in}r@{ (}r@{) \hskip 0.3in}r@{ (}r@{) \hskip 0.3in}r@{ (}r@{) \hskip 0.3in}r@{ (}r@{) \hskip 0.3in}}
   &\multicolumn{2}{c}{\itshape gic} &  \multicolumn{2}{c}{\itshape cv}& \multicolumn{2}{c}{\itshape aic}& \multicolumn{2}{c}{\itshape bic} & \multicolumn{2}{c}{\itshape gbic} & \multicolumn{2}{c}{\itshape kl oracle}\\ 
   \hline
 &\multicolumn{12}{c}{}\\
 p&\multicolumn{12}{c}{\textbf{ Kullback Leibler loss function }}\\
 &\multicolumn{12}{c}{}\\
 30 & 1.87 & 0.13 & $\bold{1.78}$ & 0.14 & 2.16 & 0.20 & 2.43 & 0.43 & 3.20 & 0.31 & 1.72 & 0.12 \\ 
   60 & 4.16 & 0.18 & $\bold{4.00}$ & 0.18 & 7.58 & 0.71 & 5.74 & 0.67 & 6.69 & 0.37 & 3.97 & 0.17 \\ 
   90 & 6.55 & 0.25 & $\bold{6.41}$ & 0.24 & 16.67 & 1.82 & 9.55 & 0.76 & 10.17 & 0.42 & 6.37 & 0.23 \\ 
   120 & 9.11 & 0.29 & $\bold{9.05}$ & 0.31 & 28.79 & 3.78 & 13.35 & 0.80 & 13.87 & 0.58 & 8.95 & 0.27 \\ 
   &\multicolumn{12}{c}{}\\
   &\multicolumn{12}{c}{\textbf{Operator norm}}\\
   &\multicolumn{12}{c}{}\\
   30  & 3.73 & 0.12 & 3.16 & 0.31 & $\bold{2.17}$ & 0.26 & 3.75 & 0.71 & 4.42 & 0.13 & 3.29 & 0.20 \\ 
   60  & 3.69 & 0.06 & 3.49 & 0.12 & $\bold{2.69}$ & 0.25 & 4.09 & 0.14 & 4.28 & 0.06 & 3.48 & 0.07 \\ 
   90  & 3.69 & 0.06 & 3.59 & 0.08 & $\bold{3.32}$ & 0.35 & 4.17 & 0.09 & 4.24 & 0.05 & 3.56 & 0.06 \\ 
   120  & 3.71 & 0.05 & $\bold{3.65}$ & 0.08 & 3.78 & 0.43 & 4.20 & 0.07 & 4.24 & 0.05 & 3.62 & 0.04 \\ 
    &\multicolumn{12}{c}{}\\
   &\multicolumn{12}{c}{\textbf{Matrix $\ell_1$ norm}}\\
   &\multicolumn{12}{c}{}\\
    30 & 4.53 & 0.15 & $\bold{4.40}$ & 0.27 & 4.55 & 0.53 & 4.58 & 0.26 & 4.76 & 0.08 & 4.42 & 0.24 \\ 
    60 & $\bold{4.48}$ & 0.12 & 4.57 & 0.15 & 7.07 & 0.86 & 4.38 & 0.06 & 4.42 & 0.03 & 4.57 & 0.14 \\ 
    90 & 4.57 & 0.12 & 4.70 & 0.16 & 9.49 & 1.06 & $\bold{4.35}$ & 0.04 & 4.36 & 0.03 & 4.73 & 0.13 \\ 
    120 & 4.68 & 0.13 & 4.78 & 0.17 & 11.52 & 1.46 & $\bold{4.34}$ & 0.03 & $\bold{4.34}$ & 0.03 & 4.83 & 0.12 \\ 
    &\multicolumn{12}{c}{}\\
    &\multicolumn{12}{c}{\textbf{Frobenius norm}}\\
    &\multicolumn{12}{c}{}\\
    30 & 6.47 & 0.17 & 5.70 & 0.34 & $\bold{4.93}$ & 0.28 & 6.72 & 0.90 & 7.78 & 0.25 & 5.82 & 0.25 \\ 
    60 & 9.15 & 0.15 & 8.69 & 0.22 & $\bold{8.59}$ & 0.42 & 10.28 & 0.42 & 10.88 & 0.20 & 8.66 & 0.12 \\ 
    90 & 11.29 & 0.15 & $\bold{10.98}$ & 0.21 & 12.75 & 0.91 & 13.05 & 0.37 & 13.35 & 0.19 & 10.90 & 0.11 \\ 
    120 & 13.16 & 0.17 & $\bold{12.95}$ & 0.28 & 16.92 & 1.56 & 15.32 & 0.32 & 15.53 & 0.22 & 12.84 & 0.11 \\ 
    &\multicolumn{12}{c}{}\\
    \hline
 \end{tabular}
 \caption{Comparison of average (SE) matrix losses for MODEL 1 which is a BAND GRAPH
 where SCAD has been used to estimate the precision matrix 
 over 100 replications \label{tab:bandgraph_norm_scad} }
 \end{table}
 \end{landscape}
 
   \renewcommand{\arraystretch}{1.2}
  \begin{landscape}
  \begin{table}[!ht]
  \centering
  \begin{tabular}{r@{\hskip 0.5in} r@{ (}r@{) \hskip 0.3in}r@{ (}r@{) \hskip 0.3in}r@{ (}r@{) \hskip 0.3in}r@{ (}r@{) \hskip 0.3in}r@{ (}r@{) \hskip 0.3in}r@{ (}r@{) \hskip 0.3in}}
    &\multicolumn{2}{c}{\itshape gic} &  \multicolumn{2}{c}{\itshape cv}& \multicolumn{2}{c}{\itshape aic}& \multicolumn{2}{c}{\itshape bic} & \multicolumn{2}{c}{\itshape gbic} & \multicolumn{2}{c}{\itshape kl oracle}\\ 
    \hline
  &\multicolumn{12}{c}{}\\
  p&\multicolumn{12}{c}{\textbf{Specificity}}\\
  &\multicolumn{12}{c}{}\\
  30 & 49.48 & 3.67 & 43.87 & 2.46 & 38.56 & 1.62 & 59.91 & 12.44 & $\bold{78.70}$ & 10.08 & 44.31 & 1.98 \\ 
  60 & 38.48 & 2.80 & 32.19 & 2.38 & 22.26 & 0.82 & 65.69 & 9.71 & $\bold{79.40}$ & 8.29 & 31.79 & 1.38 \\ 
  90 & 32.52 & 2.58 & 27.97 & 2.56 & 15.98 & 0.50 & 71.77 & 8.47 & $\bold{78.45}$ & 6.59 & 26.79 & 0.97 \\ 
  120 & 28.59 & 2.29 & 25.84 & 3.26 & 12.70 & 0.41 & 74.01 & 7.07 & $\bold{79.09}$ & 6.51 & 24.08 & 0.75 \\ 
    &\multicolumn{12}{c}{}\\
    &\multicolumn{12}{c}{\textbf{Sensitivity}}\\
    &\multicolumn{12}{c}{}\\
  30  & 60.76 & 4.31 & 76.88 & 6.02 & $\bold{92.19}$ & 2.97 & 49.71 & 22.60 & 20.49 & 9.52 & 74.55 & 5.95 \\ 
  60  & 55.98 & 3.27 & 64.43 & 3.58 & $\bold{89.08}$ & 2.47 & 28.59 & 10.40 & 14.67 & 6.06 & 64.88 & 2.59 \\ 
  90  & 54.43 & 2.90 & 59.27 & 3.33 & $\bold{86.90}$ & 2.11 & 20.53 & 7.99 & 14.25 & 4.41 & 60.55 & 2.14 \\ 
  120  & 52.88 & 2.78 & 55.86 & 3.80 & $\bold{84.28}$ & 1.97 & 16.27 & 6.50 & 12.32 & 4.62 & 57.58 & 1.97 \\ 
      &\multicolumn{12}{c}{}\\
    &\multicolumn{12}{c}{\textbf{MCC}}\\
    &\multicolumn{12}{c}{}\\
  30 & 37.22 & 3.93 & 38.07 & 3.06 & 37.73 & 3.06 & $\bold{38.04}$ & 5.37 & 31.15 & 6.96 & 37.64 & 3.34 \\ 
  60 & 36.50 & 2.25 & 33.74 & 2.00 & 28.66 & 1.64 & $\bold{36.83}$ & 6.10 & 29.99 & 6.04 & 33.55 & 1.67 \\ 
  90 & $\bold{34.77}$ & 1.97 & 32.34 & 1.98 & 24.38 & 1.18 & 34.59 & 5.15 & 30.78 & 4.50 & 31.70 & 1.33 \\ 
  120 & $\bold{32.96}$ & 1.72 & 31.42 & 2.10 & 21.72 & 1.02 & 32.00 & 5.23 & 29.11 & 4.60 & 30.49 & 1.08 \\ 
       &\multicolumn{12}{c}{}\\
     \hline
  \end{tabular}
  \caption{Comparison of average (SE) support recovery for MODEL 1which is a BAND GRAPH
  where SCAD has been used to estimate the precision matrix 
  over 100 replications \label{tab:bandgraph_scad_roc} }
  \end{table}
  \end{landscape}

  \renewcommand{\arraystretch}{1.2}
  \begin{landscape}
  \begin{table}[!ht]
  \centering
  \begin{tabular}{r@{\hskip 0.5in} r@{ (}r@{) \hskip 0.3in}r@{ (}r@{) \hskip 0.3in}r@{ (}r@{) \hskip 0.3in}r@{ (}r@{) \hskip 0.3in}r@{ (}r@{) \hskip 0.3in}r@{ (}r@{) \hskip 0.3in}}
    &\multicolumn{2}{c}{\itshape gic} &  \multicolumn{2}{c}{\itshape cv}& \multicolumn{2}{c}{\itshape aic}& \multicolumn{2}{c}{\itshape bic} & \multicolumn{2}{c}{\itshape gbic} & \multicolumn{2}{c}{\itshape kl oracle}\\ 
    \hline
  &\multicolumn{12}{c}{}\\
  p&\multicolumn{12}{c}{\textbf{ Kullback Leibler loss function }}\\
  &\multicolumn{12}{c}{}\\
  30 & 0.83 & 0.10 & 0.86 & 0.12 & $\bold{0.79}$ & 0.10 & 1.28 & 0.25 & 1.45 & 0.22 & 0.74 & 0.08 \\ 
  60 & $\bold{2.46}$ & 0.20 & 2.57 & 0.22 & 2.68 & 0.30 & 4.41 & 0.54 & 4.07 & 0.33 & 2.29 & 0.15 \\ 
  90 & $\bold{3.70}$ & 0.19 & 4.03 & 0.31 & 3.90 & 0.51 & 6.99 & 0.53 & 6.31 & 0.45 & 3.60 & 0.20 \\ 
   120 & $\bold{4.93}$ & 0.21 & 5.46 & 0.37 & 5.36 & 0.41 & 8.97 & 0.54 & 8.00 & 0.47 & 4.87 & 0.20 \\ 
    &\multicolumn{12}{c}{}\\
    &\multicolumn{12}{c}{\textbf{Operator norm}}\\
    &\multicolumn{12}{c}{}\\
     30  & 1.49 & 0.13 & 1.51 & 0.14 & $\bold{1.18}$ & 0.14 & 1.81 & 0.15 & 1.90 & 0.12 & 1.30 & 0.12 \\ 
     60  & $\bold{2.11}$ & 0.11 & 2.17 & 0.12 & 1.64 & 0.18 & 2.64 & 0.10 & 2.58 & 0.07 & 1.92 & 0.10 \\ 
     90  & 1.87 & 0.09 & 1.99 & 0.09 & $\bold{1.76}$ & 0.21 & 2.46 & 0.08 & 2.37 & 0.08 & 1.77 & 0.09 \\ 
     120  & $\bold{1.82}$ & 0.09 & 1.97 & 0.09 & 1.95 & 0.11 & 2.37 & 0.06 & 2.29 & 0.07 & 1.77 & 0.09 \\ 
      &\multicolumn{12}{c}{}\\
    &\multicolumn{12}{c}{\textbf{Matrix $\ell_1$ norm}}\\
    &\multicolumn{12}{c}{}\\
     30 & 2.48 & 0.27 & 2.51 & 0.27 & $\bold{2.29}$ & 0.29 & 2.85 & 0.26 & 2.96 & 0.23 & 2.34 & 0.28 \\ 
     60 & 4.31 & 0.31 & 4.38 & 0.32 & $\bold{3.87}$ & 0.42 & 5.20 & 0.23 & 5.08 & 0.21 & 4.06 & 0.34 \\ 
     90 & 3.91 & 0.34 & 4.05 & 0.30 & $\bold{3.87}$ & 0.39 & 4.86 & 0.22 & 4.68 & 0.24 & 3.83 & 0.36 \\ 
     120 & $\bold{3.77}$ & 0.29 & 3.83 & 0.24 & 3.82 & 0.26 & 4.29 & 0.14 & 4.15 & 0.17 & 3.76 & 0.30 \\   &\multicolumn{12}{c}{}\\
     &\multicolumn{12}{c}{\textbf{Frobenius norm}}\\
     &\multicolumn{12}{c}{}\\
    30 & 2.91 & 0.22 & 2.97 & 0.25 & $\bold{2.41}$ & 0.15 & 3.67 & 0.32 & 3.89 & 0.24 & 2.54 & 0.15 \\ 
    60 & 5.00 & 0.24 & 5.16 & 0.26 & $\bold{4.23}$ & 0.18 & 6.61 & 0.30 & 6.42 & 0.19 & 4.52 & 0.15 \\ 
      90 & 5.52 & 0.18 & 5.91 & 0.25 & $\bold{5.33}$ & 0.41 & 7.57 & 0.21 & 7.28 & 0.20 & 5.24 & 0.13 \\ 
      120 & $\bold{6.20}$ & 0.19 & 6.79 & 0.25 & 6.69 & 0.33 & 8.48 & 0.19 & 8.11 & 0.18 & 6.04 & 0.12 \\ 
     &\multicolumn{12}{c}{}\\
     \hline
  \end{tabular}
  \caption{Comparison of average (SE) matrix losses for MODEL 2 which is a SPARSE RANDOM GRAPH
  where GLASSO has been used to estimate the precision matrix 
  over 100 replications \label{tab:randomgraph} }
  \end{table}
  \end{landscape}

\renewcommand{\arraystretch}{1.2}
\begin{landscape}
\begin{table}[!ht]
\centering
\begin{tabular}{r@{\hskip 0.5in} r@{ (}r@{) \hskip 0.3in}r@{ (}r@{) \hskip 0.3in}r@{ (}r@{) \hskip 0.3in}r@{ (}r@{) \hskip 0.3in}r@{ (}r@{) \hskip 0.3in}r@{ (}r@{) \hskip 0.3in}}
  &\multicolumn{2}{c}{\itshape gic} &  \multicolumn{2}{c}{\itshape cv}& \multicolumn{2}{c}{\itshape aic}& \multicolumn{2}{c}{\itshape bic} & \multicolumn{2}{c}{\itshape gbic} & \multicolumn{2}{c}{\itshape kl oracle}\\ 
  \hline
&\multicolumn{12}{c}{}\\
p&\multicolumn{12}{c}{\textbf{Specificity}}\\
&\multicolumn{12}{c}{}\\
30 & 27.77 & 3.49 & 28.31 & 4.64 & 18.59 & 1.71 & 48.25 & 9.72 &$\bold{55.05}$ & 8.15 & 20.98 & 1.19 \\ 
60 & 23.64 & 2.44 & 25.42 & 3.66 & 13.95 & 1.62 & $\bold{64.75}$ & 10.79 & 56.56 & 6.31 & 17.93 & 0.68 \\ 
90 & 18.97 & 2.03 & 24.12 & 4.23 & 16.73 & 4.72 & $\bold{74.94}$ & 8.97 & 62.92 & 7.10 & 15.71 & 0.50 \\ 
120 & 14.62 & 1.68 & 21.96 & 4.10 & 20.96 & 4.40 & $\bold{74.20}$ & 7.19 & 59.19 & 7.15 & 12.94 & 0.35 \\ 
  &\multicolumn{12}{c}{}\\
  &\multicolumn{12}{c}{\textbf{Sensitivity}}\\
  &\multicolumn{12}{c}{}\\
30  & 99.26 & 1.55 & 99.17 & 1.59 & $\bold{99.60}$ & 1.15 & 96.40 & 3.76 & 95.14 & 4.03 & 99.57 & 1.18 \\ 
60  & 93.71 & 2.57 & 93.03 & 2.83 & $\bold{97.14}$ & 2.03 & 71.39 & 10.53 & 77.16 & 5.55 & 95.77 & 2.09 \\ 
90  & 90.73 & 2.34 & 87.32 & 3.22 & $\bold{91.89}$ & 3.86 & 47.95 & 10.36 & 60.05 & 8.25 & 92.72 & 2.13 \\ 
120  & $\bold{92.22}$ & 2.20 & 87.62 & 3.52 & 88.27 & 3.93 & 49.90 & 8.87 & 63.62 & 6.67 & 93.29 & 1.96 \\ 
     &\multicolumn{12}{c}{}\\
  &\multicolumn{12}{c}{\textbf{MCC}}\\
  &\multicolumn{12}{c}{}\\
30 & 45.89 & 4.05 & 46.39 & 5.15 & 33.63 & 2.78 & 64.16 & 7.34 & $\bold{69.15}$ & 5.60 & 37.34 & 1.82 \\ 
60 & 41.94 & 2.77 & 43.69 & 3.80 & 29.10 & 2.70 & $\bold{65.22}$ & 3.46 & 63.58 & 3.58 & 35.18 & 1.26 \\ 
90 & 37.36 & 2.50 & 42.24 & 3.84 & 34.13 & 5.68 & 57.97 & 4.95 & $\bold{59.60}$ & 3.23 & 33.56 & 1.03 \\ 
120 & 33.61 & 2.27 & 41.29 & 3.91 & 40.27 & 4.25 & 59.55 & 4.29 & $\bold{60.04}$ & 3.20 & 31.46 & 0.78 \\ 
     &\multicolumn{12}{c}{}\\
   \hline
\end{tabular}
\caption{Comparison of average (SE) support recovery for MODEL 2 which is a SPARSE RANDOM GRAPH
where GLASSO has been used to estimate the precision matrix 
over 100 replications \label{tab:randomgraph_sparse_roc} }
\end{table}
\end{landscape}

    \renewcommand{\arraystretch}{1.2}
\begin{landscape}
\begin{table}[!ht]
\centering
\begin{tabular}{r@{\hskip 0.5in} r@{ (}r@{) \hskip 0.3in}r@{ (}r@{) \hskip 0.3in}r@{ (}r@{) \hskip 0.3in}r@{ (}r@{) \hskip 0.3in}r@{ (}r@{) \hskip 0.3in}r@{ (}r@{) \hskip 0.3in}}
  &\multicolumn{2}{c}{\itshape gic} &  \multicolumn{2}{c}{\itshape cv}& \multicolumn{2}{c}{\itshape aic}& \multicolumn{2}{c}{\itshape bic} & \multicolumn{2}{c}{\itshape gbic} & \multicolumn{2}{c}{\itshape kl oracle}\\ 
  \hline
&\multicolumn{12}{c}{}\\
p&\multicolumn{12}{c}{\textbf{ Kullback Leibler loss function }}\\
&\multicolumn{12}{c}{}\\
30 & $\bold{0.68}$ & 0.11 & 1.49 & 0.30 & 0.68 & 0.12 & 2.91 & 0.60 & 0.97 & 0.18 & 0.67 & 0.11 \\ 
60 & $\bold{1.89}$ & 0.18 & 4.15 & 0.40 & 4.98 & 1.43 & 6.23 & 0.21 & 2.62 & 0.33 & 1.88 & 0.18 \\ 
90 & $\bold{2.99}$ & 0.20 & 6.30 & 0.31 & 7.89 & 0.13 & 8.55 & 0.20 & 3.94 & 0.38 & 2.96 & 0.20 \\ 
120 & $\bold{4.26}$ & 0.29 & 8.91 & 0.33 & 10.91 & 0.16 & 11.84 & 0.28 & 5.47 & 0.56 & 4.16 & 0.27 \\ 
  &\multicolumn{12}{c}{}\\
  &\multicolumn{12}{c}{\textbf{Operator norm}}\\
  &\multicolumn{12}{c}{}\\
30  & $\bold{1.50}$ & 0.18 & 2.13 & 0.18 & $\bold{1.50}$ & 0.19 & 2.62 & 0.21 & 1.81 & 0.19 & 1.48 & 0.18 \\ 
60  & $\bold{1.35}$ & 0.11 & 2.10 & 0.11 & 2.22 & 0.41 & 2.45 & 0.02 & 1.68 & 0.12 & 1.34 & 0.11 \\ 
90  & $\bold{1.34}$ & 0.09 & 2.13 & 0.06 & 2.31 & 0.01 & 2.35 & 0.01 & 1.69 & 0.09 & 1.37 & 0.08 \\ 
120  & $\bold{1.46}$ & 0.11 & 2.37 & 0.05 & 2.54 & 0.01 & 2.58 & 0.01 & 1.88 & 0.12 & 1.53 & 0.10 \\ 
     &\multicolumn{12}{c}{}\\
  &\multicolumn{12}{c}{\textbf{Matrix $\ell_1$ norm}}\\
  &\multicolumn{12}{c}{}\\
  30 & $\bold{2.31}$ & 0.38 & 3.18 & 0.36 & 2.31 & 0.39 & 3.87 & 0.30 & 2.71 & 0.38 & 2.29 & 0.38 \\ 
  60 & $\bold{2.47}$ & 0.28 & 3.54 & 0.25 & 3.81 & 0.64 & 4.14 & 0.02 & 2.84 & 0.26 & 2.46 & 0.28 \\ 
  90 & $\bold{2.60}$ & 0.28 & 3.42 & 0.10 & 3.61 & 0.01 & 3.65 & 0.01 & 2.85 & 0.25 & 2.58 & 0.27 \\ 
   120 &$\bold{2.83}$ & 0.22 & 3.54 & 0.09 & 3.74 & 0.01 & 3.78 & 0.01 & 3.01 & 0.18 & 2.81 & 0.22 \\ 
 &\multicolumn{12}{c}{}\\
   &\multicolumn{12}{c}{\textbf{Frobenius norm}}\\
   &\multicolumn{12}{c}{}\\
      30 & $\bold{2.72}$ & 0.23 & 4.01 & 0.34 & 2.72 & 0.24 & 5.24 & 0.52 & 3.33 & 0.29 & 2.69 & 0.23 \\ 
    60 & $\bold{3.65}$ & 0.17 & 5.65 & 0.25 & 6.05 & 1.13 & 6.84 & 0.09 & 4.52 & 0.27 & 3.63 & 0.18 \\ 
    90 &$\bold{4.42}$ & 0.16 & 6.95 & 0.16 & 7.74 & 0.05 & 8.00 & 0.07 & 5.54 & 0.25 & 4.49 & 0.15 \\ 
   120 & $\bold{5.21}$ & 0.20 & 8.35 & 0.15 & 9.20 & 0.06 & 9.50 & 0.09 & 6.58 & 0.33 & 5.35 & 0.18 \\ 
    &\multicolumn{12}{c}{}\\
   \hline
\end{tabular}
\caption{Comparison of average (SE) matrix losses for MODEL 2 which is a SPARSE RANDOM GRAPH 
where ADAPTIVE LASSO has been used to estimate the precision matrix 
over 100 replications \label{tab:randomgraph_adptive_norm} }
\end{table}
\end{landscape}

\renewcommand{\arraystretch}{1.2}
\begin{landscape}
\begin{table}[!ht]
\centering
\begin{tabular}{r@{\hskip 0.5in} r@{ (}r@{) \hskip 0.3in}r@{ (}r@{) \hskip 0.3in}r@{ (}r@{) \hskip 0.3in}r@{ (}r@{) \hskip 0.3in}r@{ (}r@{) \hskip 0.3in}r@{ (}r@{) \hskip 0.3in}}
  &\multicolumn{2}{c}{\itshape gic} &  \multicolumn{2}{c}{\itshape cv}& \multicolumn{2}{c}{\itshape aic}& \multicolumn{2}{c}{\itshape bic} & \multicolumn{2}{c}{\itshape gbic} & \multicolumn{2}{c}{\itshape kl oracle}\\ 
  \hline
&\multicolumn{12}{c}{}\\
p&\multicolumn{12}{c}{\textbf{Specificity}}\\
&\multicolumn{12}{c}{}\\
 30& 54.14 & 6.96 & $\bold{90.47}$ & 7.82 & 52.81 & 7.77 & 24.96 & 42.52 & 78.07 & 8.02 & 51.62 & 5.87 \\ 
  60& 39.15 & 4.32 & $\bold{94.66}$ & 5.35 & 10.43 & 23.60 & 0.00 & 0.00 & 70.29 & 7.90 & 37.94 & 3.99 \\ 
 90 & 27.68 & 3.02 & $\bold{97.48}$ & 3.27 & 0.00 & 0.00 & 0.00 & 0.00 & 66.73 & 8.57 & 30.16 & 1.73 \\ 
 120 & 24.26 & 2.76 & $\bold{98.08}$ & 2.80 & 1.00 & 10.00 & 0.00 & 0.00 & 64.45 & 8.33 & 28.85 & 1.46 \\ 
  &\multicolumn{12}{c}{}\\
  &\multicolumn{12}{c}{\textbf{Sensitivity}}\\
  &\multicolumn{12}{c}{}\\
 30  & 96.70 & 3.23 & 76.39 & 11.44 & $\bold{96.85}$ & 3.55 & 12.06 & 26.16 & 89.24 & 6.31 & 97.06 & 3.15 \\ 
  60 & $\bold{85.55}$ & 4.07 & 30.54 & 12.55 & 13.91 & 31.95 & 0.00 & 0.00 & 67.32 & 8.19 & 86.15 & 4.07 \\ 
   90& $\bold{85.20}$ & 3.08 & 16.37 & 7.38 & 0.00 & 0.00 & 0.00 & 0.00 & 62.92 & 7.74 & 83.81 & 2.99 \\ 
   120& $\bold{87.22}$ & 3.04 & 13.58 & 6.17 & 0.01 & 0.06 & 0.00 & 0.00 & 64.32 & 8.18 & 84.93 & 3.02 \\ 
     &\multicolumn{12}{c}{}\\
  &\multicolumn{12}{c}{\textbf{MCC}}\\
  &\multicolumn{12}{c}{}\\
 30 & 69.41 & 5.15 & 81.42 & 5.66 & 68.40 & 5.33 & 15.38 & 29.20 & $\bold{81.80}$ & 4.68 & 67.68 & 4.48 \\ 
 60 & 54.55 & 3.84 & 51.33 & 9.81 & 9.34 & 20.54 & 0.00 & 0.00 & $\bold{66.71}$ & 4.16 & 53.77 & 3.55 \\ 
 90 & 45.83 & 2.88 & 38.18 & 8.94 & 0.00 & 0.00 & 0.00 & 0.00 &$\bold{63.21}$ & 3.63 & 47.79 & 2.01 \\ 
 120 & 43.77 & 2.68 & 35.16 & 8.11 & 0.08 & 0.75 & 0.00 & 0.00 &$\bold{63.03}$ & 3.47 & 47.61 & 1.77 \\ 
       &\multicolumn{12}{c}{}\\
   \hline
\end{tabular}
\caption{Comparison of average (SE) support recovery for MODEL 2 which is a SPARSE RANDOM GRAPH
where  ADAPTIVE LASSO has been used to estimate the precision matrix 
over 100 replications \label{tab:bandgraph_adaptive_roc_1} }
\end{table}
\end{landscape}

\renewcommand{\arraystretch}{1.2}

\begin{landscape}
\begin{table}[!ht]
\centering
\begin{tabular}{r@{\hskip 0.5in} r@{ (}r@{) \hskip 0.3in}r@{ (}r@{) \hskip 0.3in}r@{ (}r@{) \hskip 0.3in}r@{ (}r@{) \hskip 0.3in}r@{ (}r@{) \hskip 0.3in}r@{ (}r@{) \hskip 0.3in}}
  &\multicolumn{2}{c}{\itshape gic} &  \multicolumn{2}{c}{\itshape cv}& \multicolumn{2}{c}{\itshape aic}& \multicolumn{2}{c}{\itshape bic} & \multicolumn{2}{c}{\itshape gbic} & \multicolumn{2}{c}{\itshape kl oracle}\\ 
  \hline
&\multicolumn{12}{c}{}\\
p&\multicolumn{12}{c}{\textbf{ Kullback Leibler loss function}}\\
&\multicolumn{12}{c}{}\\
30 & $\bold{0.79}$ & 0.12 & $\bold{0.75}$ & 0.14 & 1.20 & 0.20 & 1.04 & 0.21 & 1.47 & 0.26 & 0.70 & 0.10 \\ 
60 & 1.93 & 0.19 & $\bold{1.86}$ & 0.17 & 5.09 & 0.69 & 2.74 & 0.32 & 3.51 & 0.36 & 1.81 & 0.15 \\ 
90 & 3.05 & 0.21 & $\bold{3.01}$ & 0.22 & 12.84 & 1.87 & 4.26 & 0.50 & 5.29 & 0.48 & 2.96 & 0.20 \\ 
120 & $\bold{4.16}$ & 0.24 & $\bold{4.16}$ & 0.26 & 22.46 & 4.31 & 6.08 & 0.66 & 7.10 & 0.55 & 4.10 & 0.23 \\ 
  &\multicolumn{12}{c}{}\\
  &\multicolumn{12}{c}{\textbf{Operator norm}}\\
  &\multicolumn{12}{c}{}\\
 30  & 1.55 & 0.16 & $\bold{1.24}$ & 0.25 & 1.38 & 0.28 & 1.79 & 0.18 & 2.07 & 0.14 & 1.27 & 0.19 \\ 
 60  & 1.52 & 0.12 & $\bold{1.39}$ & 0.16 & 2.07 & 0.24 & 1.83 & 0.10 & 2.03 & 0.10 & 1.36 & 0.11 \\ 
 90  & 1.94 & 0.16 & $\bold{1.85}$ & 0.16 & 2.73 & 0.29 & 2.34 & 0.15 & 2.59 & 0.12 & 1.82 & 0.15 \\ 
 120  & 1.53 & 0.10 & $\bold{1.51}$ & 0.10 & 3.15 & 0.43 & 1.94 & 0.11 & 2.09 & 0.08 & 1.49 & 0.08 \\ 
     &\multicolumn{12}{c}{}\\
  &\multicolumn{12}{c}{\textbf{Matrix $\ell_1$ norm}}\\
  &\multicolumn{12}{c}{}\\
 30 & 2.54 & 0.39 & $\bold{2.22}$ & 0.42 & 2.80 & 0.53 & 2.89 & 0.37 & 3.30 & 0.29 & 2.22 & 0.41 \\ 
 60 & 2.78 & 0.26 & $\bold{2.66}$ & 0.32 & 4.95 & 0.73 & 3.15 & 0.24 & 3.47 & 0.22 & 2.63 & 0.30 \\ 
 90 & 3.90 & 0.42 & $\bold{3.76}$ & 0.42 & 7.43 & 0.93 & 4.63 & 0.40 & 5.12 & 0.30 & 3.71 & 0.43 \\ 
 120 &$\bold{2.90}$ & 0.26 & $\bold{2.90}$ & 0.26 & 8.92 & 1.38 & 3.18 & 0.23 & 3.38 & 0.19 & 2.90 & 0.26 \\ 
  &\multicolumn{12}{c}{}\\
   &\multicolumn{12}{c}{\textbf{Frobenius norm}}\\
   &\multicolumn{12}{c}{}\\
  30 & 2.81 & 0.24 & $\bold{2.49}$ & 0.28 & 3.03 & 0.35 & 3.26 & 0.32 & 3.88 & 0.30 & 2.44 & 0.19 \\ 
  60 & 3.91 & 0.21 & $\bold{3.66}$ & 0.25 & 6.27 & 0.61 & 4.77 & 0.26 & 5.47 & 0.27 & 3.58 & 0.16 \\ 
  90 & 4.86 & 0.21 & $\bold{4.71}$ & 0.21 & 10.35 & 1.05 & 5.94 & 0.36 & 6.68 & 0.29 & 4.66 & 0.16 \\ 
   120 & 5.31 & 0.20 & $\bold{5.27}$ & 0.21 & 13.99 & 2.02 & 6.79 & 0.39 & 7.39 & 0.28 & 5.19 & 0.15 \\ 
        &\multicolumn{12}{c}{}\\
   \hline
\end{tabular}
\caption{Comparison of average (SE) matrix losses for MODEL 2 which is a SPARSE RANDOM GRAPH
where SCAD has been used to estimate the precision matrix 
over 100 replications \label{tab:randomgraph_scad_norm} }
\end{table}
\end{landscape}

 \renewcommand{\arraystretch}{1.2}
 \begin{landscape}
 \begin{table}[!ht]
 \centering
 \begin{tabular}{r@{\hskip 0.5in} r@{ (}r@{) \hskip 0.3in}r@{ (}r@{) \hskip 0.3in}r@{ (}r@{) \hskip 0.3in}r@{ (}r@{) \hskip 0.3in}r@{ (}r@{) \hskip 0.3in}r@{ (}r@{) \hskip 0.3in}}
   &\multicolumn{2}{c}{\itshape gic} &  \multicolumn{2}{c}{\itshape cv}& \multicolumn{2}{c}{\itshape aic}& \multicolumn{2}{c}{\itshape bic} & \multicolumn{2}{c}{\itshape gbic} & \multicolumn{2}{c}{\itshape kl oracle}\\ 
   \hline
 &\multicolumn{12}{c}{}\\
 p&\multicolumn{12}{c}{\textbf{Specificity}}\\
 &\multicolumn{12}{c}{}\\
   30 & 40.72 & 4.44 & 29.11 & 4.41 & 17.83 & 1.53 & 55.06 & 9.24 & $\bold{73.95}$ & 9.46 & 30.52 & 3.30 \\ 
   60 & 32.40 & 3.86 & 25.62 & 3.88 & 10.46 & 1.75 & 60.06 & 6.10 & $\bold{76.52}$ & 6.29 & 24.05 & 1.49 \\ 
   90 & 26.73 & 4.14 & 22.54 & 3.87 & 6.87 & 0.36 & 58.96 & 7.40 & $\bold{77.04}$ & 6.42 & 21.12 & 0.93 \\ 
   120 & 23.19 & 3.58 & 22.03 & 4.31 & 5.80 & 2.78 & 62.41 & 7.77 & $\bold{76.83}$ & 6.01 & 20.04 & 0.80 \\ 
   &\multicolumn{12}{c}{}\\
   &\multicolumn{12}{c}{\textbf{Sensitivity}}\\
   &\multicolumn{12}{c}{}\\
        30  & 94.66 & 3.50 & 97.53 & 2.36 & $\bold{98.89}$ & 1.41 & 89.42 & 5.78 & 78.50 & 7.24 & 97.42 & 2.42 \\ 
        60  & 93.00 & 3.23 & 95.29 & 2.68 & $\bold{98.01}$ & 1.65 & 80.57 & 7.08 & 68.44 & 7.89 & 95.93 & 2.15 \\ 
        90  & 88.93 & 3.58 & 91.06 & 3.05 & $\bold{96.71}$ & 1.67 & 70.35 & 7.54 & 55.71 & 8.38 & 91.66 & 2.75 \\ 
          120  & 87.73 & 3.38 & 88.70 & 3.34 & $\bold{95.41}$ & 1.95 & 63.16 & 8.64 & 51.30 & 8.22 & 89.80 & 2.44 \\ 
         &\multicolumn{12}{c}{}\\
   &\multicolumn{12}{c}{\textbf{MCC}}\\
   &\multicolumn{12}{c}{}\\
       30 & 57.03 & 4.01 & 46.13 & 4.95 & 31.11 & 2.77 & 66.22 & 5.93 & $\bold{73.46}$ & 5.04 & 47.83 & 3.43 \\ 
       60 & 51.37 & 3.31 & 45.06 & 3.95 & 23.53 & 2.62 & 67.54 & 3.59 & $\bold{70.76}$ & 4.17 & 43.68 & 1.78 \\ 
       90 & 45.91 & 3.63 & 42.08 & 3.87 & 18.80 & 1.12 & 62.73 & 2.89 & $\bold{64.18}$ & 3.80 & 40.85 & 1.36 \\ 
       120 & 42.72 & 3.22 & 41.65 & 3.95 & 17.39 & 3.93 & 61.35 & 3.00 & $\bold{61.66}$ & 4.14 & 39.99 & 1.21\\        &\multicolumn{12}{c}{}\\
    \hline
 \end{tabular}
 \caption{Comparison of average (SE) support recovery for MODEL 2 which is a SPARSE RANDOM GRAPH where SCAD has been used to estimate the precision matrix 
 over 100 replications \label{tab:bandgraph_scad_roc_2} }
 \end{table}
 \end{landscape}
 
   \renewcommand{\arraystretch}{1.2}
   \begin{landscape}
   \begin{table}[!ht]
   \centering
   \begin{tabular}{r@{\hskip 0.5in} r@{ (}r@{) \hskip 0.3in}r@{ (}r@{) \hskip 0.3in}r@{ (}r@{) \hskip 0.3in}r@{ (}r@{) \hskip 0.3in}r@{ (}r@{) \hskip 0.3in}r@{ (}r@{) \hskip 0.3in}}
     &\multicolumn{2}{c}{\itshape gic} &  \multicolumn{2}{c}{\itshape cv}& \multicolumn{2}{c}{\itshape aic}& \multicolumn{2}{c}{\itshape bic} & \multicolumn{2}{c}{\itshape gbic} & \multicolumn{2}{c}{\itshape kl oracle}\\ 
     \hline
   &\multicolumn{12}{c}{}\\
   p&\multicolumn{12}{c}{\textbf{ Kullback Leibler loss function}}\\
   &\multicolumn{12}{c}{}\\
   30 & $\bold{1.71}$ & 0.10 & 1.74 & 0.12 & 1.66 & 0.10 & 2.91 & 0.32 & 2.44 & 0.28 & 1.61 & 0.08 \\ 
   60 & $\bold{3.92}$ & 0.18 & 4.13 & 0.22 & 4.22 & 0.23 & 7.83 & 0.38 & 5.84 & 0.55 & 3.82 & 0.15 \\ 
   90 & $\bold{6.90}$ & 0.21 & 7.32 & 0.31 & 8.52 & 0.63 & 11.85 & 0.22 & 10.63 & 0.58 & 6.82 & 0.19 \\ 
   120 & $\bold{9.74}$ & 0.20 & 10.37 & 0.35 & 13.13 & 1.33 & 16.21 & 0.24 & 15.14 & 0.58 & 9.70 & 0.20 \\ 
     &\multicolumn{12}{c}{}\\
     &\multicolumn{12}{c}{\textbf{Operator norm}}\\
     &\multicolumn{12}{c}{}\\
     30  & 10.20 & 0.07 & 10.22 & 0.12 & $\bold{9.71}$ & 0.11 & 10.58 & 0.03 & 10.51 & 0.05 & 9.94 & 0.06 \\ 
     60  & 15.21 & 0.05 & 15.29 & 0.06 & $\bold{14.69}$ & 0.11 & 15.52 & 0.01 & 15.48 & 0.02 & 15.07 & 0.04 \\ 
     90  & 17.79 & 0.04 & 17.90 & 0.05 & $\bold{17.16}$ & 0.14 & 18.08 & 0.01 & 18.06 & 0.01 & 17.72 & 0.04 \\ 
     120  & 21.76 & 0.04 & 21.89 & 0.04 & $\bold{21.13}$ & 0.20 & 22.04 & 0.01 & 22.03 & 0.01 & 21.73 & 0.03 \\ 
        &\multicolumn{12}{c}{}\\
     &\multicolumn{12}{c}{\textbf{Matrix $\ell_1$ norm}}\\
     &\multicolumn{12}{c}{}\\
     30 & 11.00 & 0.14 & 11.03 & 0.17 & $\bold{10.58}$ & 0.18 & 11.34 & 0.07 & 11.31 & 0.09 & 10.77 & 0.15 \\ 
     60 & 16.91 & 0.20 & 16.99 & 0.19 & $\bold{16.37}$ & 0.28 & 17.19 & 0.06 & 17.19 & 0.14 & 16.76 & 0.22 \\ 
     90 & 19.92 & 0.17 & 20.03 & 0.17 & $\bold{19.33}$ & 0.28 & 20.14 & 0.07 & 20.19 & 0.10 & 19.84 & 0.19 \\ 
     120 & 24.48 & 0.22 & 24.60 & 0.20 & $\bold{23.95}$ & 0.34 & 24.70 & 0.06 & 24.71 & 0.11 & 24.46 & 0.23 \\ 
     &\multicolumn{12}{c}{}\\
      &\multicolumn{12}{c}{\textbf{Frobenius norm}}\\
      &\multicolumn{12}{c}{}\\
     30 & 10.59 & 0.10 & 10.63 & 0.16 & $\bold{10.06}$ & 0.10 & 11.33 & 0.11 & 11.15 & 0.12 & 10.28 & 0.07 \\ 
     60 & 15.91 & 0.08 & 16.07 & 0.10 & $\bold{15.33}$ & 0.08 & 17.03 & 0.07 & 16.65 & 0.11 & 15.70 & 0.04 \\ 
     90 & 18.93 & 0.07 & 19.20 & 0.12 & $\bold{18.31}$ & 0.06 & 20.19 & 0.04 & 19.98 & 0.09 & 18.79 & 0.04 \\ 
     120 & 23.06 & 0.07 & 23.43 & 0.11 & $\bold{22.59}$ & 0.15 & 24.43 & 0.04 & 24.27 & 0.08 & 23.01 & 0.04 \\
           &\multicolumn{12}{c}{}\\
      \hline
   \end{tabular}
   \caption{Comparison of average (SE) matrix losses for MODEL 3 which is a DENSE RANDOM  GRAPH 
   where GLASSO has been used to estimate the precision matrix 
   over 100 replications \label{tab:dense_randomgraph_glasso_norm} }
   \end{table}
   \end{landscape}
   \renewcommand{\arraystretch}{1.2}
\begin{landscape}
\begin{table}[!ht]
\centering
\begin{tabular}{r@{\hskip 0.5in} r@{ (}r@{) \hskip 0.3in}r@{ (}r@{) \hskip 0.3in}r@{ (}r@{) \hskip 0.3in}r@{ (}r@{) \hskip 0.3in}r@{ (}r@{) \hskip 0.3in}r@{ (}r@{) \hskip 0.3in}}
  &\multicolumn{2}{c}{\itshape gic} &  \multicolumn{2}{c}{\itshape cv}& \multicolumn{2}{c}{\itshape aic}& \multicolumn{2}{c}{\itshape bic} & \multicolumn{2}{c}{\itshape gbic} & \multicolumn{2}{c}{\itshape kl oracle}\\ 
  \hline
&\multicolumn{12}{c}{}\\
p&\multicolumn{12}{c}{\textbf{Specificity}}\\
&\multicolumn{12}{c}{}\\
    & N/A & N/A & N/A & N/A& N/A & N/A& N/A&N/A & N/A &N/A & N/A & N/A\\ 
  & N/A & N/A & N/A & N/A& N/A & N/A& N/A&N/A & N/A &N/A & N/A & N/A\\ 
   & N/A & N/A & N/A & N/A& N/A & N/A& N/A&N/A & N/A &N/A & N/A & N/A\\ 
    & N/A & N/A & N/A & N/A& N/A & N/A& N/A&N/A & N/A &N/A & N/A & N/A\\   &\multicolumn{12}{c}{}\\
  &\multicolumn{12}{c}{\textbf{Sensitivity}}\\
  &\multicolumn{12}{c}{}\\
30  & 33.10 & 3.68 & 32.70 & 6.19 & $\bold{51.09}$ & 3.61 & 4.35 & 3.62 & 11.04 & 4.65 & 44.91 & 3.35 \\ 
60  & 25.97 & 1.80 & 22.30 & 2.53 & $\bold{41.97}$ & 2.63 & 0.59 & 0.56 & 7.26 & 2.67 & 31.79 & 1.48 \\ 
90  & 22.67 & 1.56 & 17.16 & 2.69 & $\bold{40.98}$ & 2.85 & 0.25 & 0.15 & 1.94 & 1.22 & 26.10 & 1.03 \\ 
120  & 20.79 & 1.44 & 14.18 & 2.00 & $\bold{37.13}$ & 5.00 & 0.17 & 0.10 & 0.92 & 0.64 & 22.00 & 0.77 \\ 
         &\multicolumn{12}{c}{}\\
  &\multicolumn{12}{c}{\textbf{MCC}}\\
  &\multicolumn{12}{c}{}\\
    & N/A & N/A & N/A & N/A& N/A & N/A& N/A&N/A & N/A &N/A & N/A & N/A\\ 
  & N/A & N/A & N/A & N/A& N/A & N/A& N/A&N/A & N/A &N/A & N/A & N/A\\ 
   & N/A & N/A & N/A & N/A& N/A & N/A& N/A&N/A & N/A &N/A & N/A & N/A\\ 
    & N/A & N/A & N/A & N/A& N/A & N/A& N/A&N/A & N/A &N/A & N/A & N/A\\ &\multicolumn{12}{c}{}\\
   \hline
\end{tabular}
\caption{Comparison of average (SE) support recovery for MODEL 3 which is a DENSE RANDOM GRAPH where GLASSO has been used to estimate the precision matrix 
over 100 replications \label{tab:dense_random_graph_roc} }
\end{table}
\end{landscape}

   \renewcommand{\arraystretch}{1.2}
\begin{landscape}
\begin{table}[!ht]
\centering
\begin{tabular}{r@{\hskip 0.5in} r@{ (}r@{) \hskip 0.3in}r@{ (}r@{) \hskip 0.3in}r@{ (}r@{) \hskip 0.3in}r@{ (}r@{) \hskip 0.3in}r@{ (}r@{) \hskip 0.3in}r@{ (}r@{) \hskip 0.3in}}
  &\multicolumn{2}{c}{\itshape gic} &  \multicolumn{2}{c}{\itshape cv}& \multicolumn{2}{c}{\itshape aic}& \multicolumn{2}{c}{\itshape bic} & \multicolumn{2}{c}{\itshape gbic} & \multicolumn{2}{c}{\itshape kl oracle}\\ 
  \hline
&\multicolumn{12}{c}{}\\
p&\multicolumn{12}{c}{\textbf{ Kullback Leibler loss function }}\\
&\multicolumn{12}{c}{}\\
30 & $\bold{1.82}$ & 0.10 & 2.56 & 0.23 & $\bold{1.82}$ & 0.10 & 4.05 & 0.11 & 2.13 & 0.21 & 1.81 & 0.10 \\ 
60 & $\bold{4.52}$ & 0.17 & 6.86 & 0.36 & 5.58 & 1.96 & 9.56 & 0.17 & 5.10 & 0.36 & 4.52 & 0.17 \\ 
90 & $\bold{7.85}$ & 0.22 & 11.24 & 0.29 & 13.60 & 0.13 & 14.14 & 0.20 & 8.92 & 0.44 & 7.84 & 0.22 \\ 
120 & $\bold{11.30}$ & 0.29 & 15.63 & 0.27 & 18.32 & 0.16 & 19.03 & 0.23 & 12.70 & 0.57 & 11.30 & 0.28 \\ 
  &\multicolumn{12}{c}{}\\
  &\multicolumn{12}{c}{\textbf{Operator norm}}\\
  &\multicolumn{12}{c}{}\\
30  & 9.10 & 0.07 & 9.30 & 0.04 & $\bold{9.09}$ & 0.07 & 9.48 & 0.02 & 9.24 & 0.05 & 9.09 & 0.07 \\ 
  60  & $\bold{17.04}$ & 0.05 & 17.12 & 0.02 & 17.08 & 0.08 & 17.24 & 0.01 & 17.13 & 0.04 & 17.04 & 0.05 \\ 
  90  & $\bold{20.48}$ & 0.05 & 20.58 & 0.01 & 20.70 & 0.01 & 20.73 & 0.01 & 20.60 & 0.03 & 20.48 & 0.05 \\ 
  120  & $\bold{23.07}$ & 0.06 & 23.17 & 0.01 & 23.30 & 0.01 & 23.33 & 0.01 & 23.19 & 0.03 & 23.07 & 0.05 \\ 
     &\multicolumn{12}{c}{}\\
  &\multicolumn{12}{c}{\textbf{Matrix $\ell_1$ norm}}\\
  &\multicolumn{12}{c}{}\\
   30 & 9.87 & 0.17 & 10.01 & 0.14 & $\bold{9.86}$ & 0.17 & 10.12 & 0.02 & 9.98 & 0.16 & 9.86 & 0.17 \\ 
   60 & 19.52 & 0.33 & 19.57 & 0.21 & $\bold{19.51}$ & 0.29 & 19.55 & 0.01 & 19.62 & 0.30 & 19.52 & 0.33 \\ 
   90 & 23.27 & 0.26 & $\bold{23.20}$ & 0.11 & 23.25 & 0.01 & 23.28 & 0.01 & 23.36 & 0.21 & 23.27 & 0.26 \\ 
   120 & 26.23 & 0.30 & $\bold{26.09}$ & 0.07 & 26.18 & 0.01 & 26.21 & 0.01 & 26.28 & 0.23 & 26.23 & 0.30 \\ 
  &\multicolumn{12}{c}{}\\
   &\multicolumn{12}{c}{\textbf{Frobenius norm}}\\
   &\multicolumn{12}{c}{}\\
  30 & 9.51 & 0.07 & 9.94 & 0.08 & $\bold{9.50}$ & 0.07 & 10.57 & 0.04 & 9.75 & 0.09 & 9.50 & 0.07 \\ 
  60 & $\bold{17.80}$ & 0.07 & 18.41 & 0.06 & 18.07 & 0.51 & 19.10 & 0.04 & 18.04 & 0.09 & 17.80 & 0.07 \\ 
  90 & $\bold{21.61}$ & 0.07 & 22.41 & 0.05 & 23.01 & 0.03 & 23.12 & 0.04 & 22.01 & 0.09 & 21.61 & 0.07 \\ 
  120 & $\bold{24.47}$ & 0.09 & 25.38 & 0.05 & 26.01 & 0.03 & 26.13 & 0.04 & 24.93 & 0.09 & 24.48 & 0.08 \\          &\multicolumn{12}{c}{}\\
   \hline
\end{tabular}
\caption{Comparison of average (SE) matrix losses for MODEL 3 which is a DENSE RANDOM  GRAPH 
where ADAPTIVE LASSO has been used to estimate the precision matrix 
over 100 replications \label{tab:randomgraph_adaptive_norm} }
\end{table}
\end{landscape}

\renewcommand{\arraystretch}{1.2}
\begin{landscape}
\begin{table}[!ht]
\centering
\begin{tabular}{r@{\hskip 0.5in} r@{ (}r@{) \hskip 0.3in}r@{ (}r@{) \hskip 0.3in}r@{ (}r@{) \hskip 0.3in}r@{ (}r@{) \hskip 0.3in}r@{ (}r@{) \hskip 0.3in}r@{ (}r@{) \hskip 0.3in}}
  &\multicolumn{2}{c}{\itshape gic} &  \multicolumn{2}{c}{\itshape cv}& \multicolumn{2}{c}{\itshape aic}& \multicolumn{2}{c}{\itshape bic} & \multicolumn{2}{c}{\itshape gbic} & \multicolumn{2}{c}{\itshape kl oracle}\\ 
  \hline
&\multicolumn{12}{c}{}\\
p&\multicolumn{12}{c}{\textbf{Specificity}}\\
&\multicolumn{12}{c}{}\\
    & N/A & N/A & N/A & N/A& N/A & N/A& N/A&N/A & N/A &N/A & N/A & N/A\\ 
  & N/A & N/A & N/A & N/A& N/A & N/A& N/A&N/A & N/A &N/A & N/A & N/A\\ 
   & N/A & N/A & N/A & N/A& N/A & N/A& N/A&N/A & N/A &N/A & N/A & N/A\\ 
    & N/A & N/A & N/A & N/A& N/A & N/A& N/A&N/A & N/A &N/A & N/A & N/A\\    &\multicolumn{12}{c}{}\\
  &\multicolumn{12}{c}{\textbf{Sensitivity}}\\
  &\multicolumn{12}{c}{}\\
   30  & 17.23 & 2.71 & 4.24 & 2.22 & $\bold{17.93}$ & 2.96 & 0.00 & 0.00 & 8.41 & 2.58 & 18.03 & 2.89 \\ 
    60  & $\bold{10.96}$ & 1.48 & 0.86 & 0.56 & 8.30 & 4.85 & 0.00 & 0.00 & 5.73 & 1.34 & 11.06 & 1.49 \\ 
    90  & $\bold{10.84}$ & 1.25 & 0.32 & 0.20 & 0.00 & 0.00 & 0.00 & 0.00 & 4.11 & 1.12 & 10.95 & 1.29 \\ 
    120  & $\bold{9.83}$ & 1.52 & 0.12 & 0.09 & 0.00 & 0.00 & 0.00 & 0.00 & 3.08 & 0.89 & 9.63 & 1.31 \\ 
          &\multicolumn{12}{c}{}\\
  &\multicolumn{12}{c}{\textbf{MCC}}\\
  &\multicolumn{12}{c}{}\\
     & N/A & N/A & N/A & N/A& N/A & N/A& N/A&N/A & N/A &N/A & N/A & N/A\\ 
   & N/A & N/A & N/A & N/A& N/A & N/A& N/A&N/A & N/A &N/A & N/A & N/A\\ 
    & N/A & N/A & N/A & N/A& N/A & N/A& N/A&N/A & N/A &N/A & N/A & N/A\\ 
     & N/A & N/A & N/A & N/A& N/A & N/A& N/A&N/A & N/A &N/A & N/A & N/A\\ &\multicolumn{12}{c}{}\\
   \hline
\end{tabular}
\caption{Comparison of average (SE) support recovery for MODEL 3 which is a DENSE RANDOM GRAPH where ADAPTIVE LASSO has been used to estimate the precision matrix 
over 100 replications \label{tab:adaptive_random_roc} }
\end{table}
\end{landscape}

  \renewcommand{\arraystretch}{1.2}
  \begin{landscape}
  \begin{table}[!ht]
  \centering
  \begin{tabular}{r@{\hskip 0.5in} r@{ (}r@{) \hskip 0.3in}r@{ (}r@{) \hskip 0.3in}r@{ (}r@{) \hskip 0.3in}r@{ (}r@{) \hskip 0.3in}r@{ (}r@{) \hskip 0.3in}r@{ (}r@{) \hskip 0.3in}}
    &\multicolumn{2}{c}{\itshape gic} &  \multicolumn{2}{c}{\itshape cv}& \multicolumn{2}{c}{\itshape aic}& \multicolumn{2}{c}{\itshape bic} & \multicolumn{2}{c}{\itshape gbic} & \multicolumn{2}{c}{\itshape kl oracle}\\ 
    \hline
  &\multicolumn{12}{c}{}\\
  p&\multicolumn{12}{c}{\textbf{ Kullback Leibler loss function}}\\
  &\multicolumn{12}{c}{}\\
   30 & $\bold{1.79}$ & 0.11 & 1.83 & 0.14 & 2.42 & 0.22 & 2.42 & 0.30 & 2.49 & 0.36 & 1.77 & 0.10 \\ 
   60 & $\bold{4.10}$ & 0.15 & $\bold{4.10}$ & 0.16 & 7.85 & 0.75 & 6.26 & 0.57 & 6.13 & 0.48 & 4.07 & 0.15 \\ 
   90 & $\bold{7.03}$ & 0.21 & 7.04 & 0.23 & 15.58 & 1.97 & 11.05 & 0.56 & 10.44 & 0.66 & 6.97 & 0.20 \\ 
   120 & $\bold{9.60}$ & 0.24 & 9.63 & 0.27 & 32.62 & 4.23 & 13.99 & 0.32 & 13.72 & 0.43 & 9.57 & 0.24 \\ 
    &\multicolumn{12}{c}{}\\
    &\multicolumn{12}{c}{\textbf{Operator norm}}\\
    &\multicolumn{12}{c}{}\\
   30  & 10.74 & 0.06 & 10.57 & 0.16 & $\bold{9.64}$ & 0.26 & 10.94 & 0.09 & 10.96 & 0.05 & 10.70 & 0.07 \\ 
   60  & 15.16 & 0.05 & 15.13 & 0.07 & $\bold{13.98}$ & 0.24 & 15.36 & 0.02 & 15.35 & 0.02 & 15.14 & 0.04 \\ 
   90  & 20.02 & 0.04 & 20.02 & 0.06 & $\bold{18.77}$ & 0.29 & 20.22 & 0.02 & 20.20 & 0.02 & 19.99 & 0.04 \\ 
   120  & 20.69 & 0.04 & 20.70 & 0.04 & $\bold{18.94}$ & 0.30 & 20.90 & 0.01 & 20.89 & 0.01 & 20.69 & 0.03 \\ 
       &\multicolumn{12}{c}{}\\
    &\multicolumn{12}{c}{\textbf{Matrix $\ell_1$ norm}}\\
    &\multicolumn{12}{c}{}\\
     30 & 11.76 & 0.18 & 11.65 & 0.24 & $\bold{11.08}$ & 0.41 & 11.92 & 0.15 & 11.93 & 0.12 & 11.72 & 0.19 \\ 
     60 & 17.59 & 0.26 & 17.57 & 0.27 & $\bold{17.54}$ & 0.74 & 17.79 & 0.13 & 17.80 & 0.13 & 17.57 & 0.26 \\ 
     90 & 23.72 & 0.27 & $\bold{23.71}$ & 0.27 & 24.18 & 0.86 & 23.97 & 0.12 & 23.97 & 0.13 & 23.68 & 0.28 \\ 
     120 & $\bold{23.18}$ & 0.21 & 23.19 & 0.22 & 26.78 & 1.12 & 23.31 & 0.07 & 23.31 & 0.08 & 23.18 & 0.21 \\ 
     &\multicolumn{12}{c}{}\\
     &\multicolumn{12}{c}{\textbf{Frobenius norm}}\\
     &\multicolumn{12}{c}{}\\
     30 & 11.05 & 0.07 & 10.91 & 0.12 & $\bold{10.44}$ & 0.17 & 11.45 & 0.15 & 11.49 & 0.14 & 11.00 & 0.06 \\ 
      60 & 15.80 & 0.07 & $\bold{15.77}$ & 0.08 & 16.03 & 0.17 & 16.54 & 0.16 & 16.50 & 0.13 & 15.77 & 0.04 \\ 
     90 & 21.06 & 0.08 & $\bold{21.05}$ & 0.09 & 22.03 & 0.48 & 22.05 & 0.11 & 21.91 & 0.13 & 20.99 & 0.04 \\ 
      120 & $\bold{21.89}$ & 0.05 & 21.91 & 0.07 & 26.27 & 1.27 & 22.92 & 0.07 & 22.86 & 0.09 & 21.89 & 0.03 \\ &\multicolumn{12}{c}{}\\
     \hline
  \end{tabular}
  \caption{Comparison of average (SE) matrix losses for MODEL 3 which is a DENSE RANDOM  GRAPH 
  where  SCAD has been used to estimate the precision matrix 
  over 100 replications \label{tab:randomgraph_scad_norm_1} }
  \end{table}
  \end{landscape}

 \renewcommand{\arraystretch}{1.2}
 \begin{landscape}
 \begin{table}[!ht]
 \centering
 \begin{tabular}{r@{\hskip 0.5in} r@{ (}r@{) \hskip 0.3in}r@{ (}r@{) \hskip 0.3in}r@{ (}r@{) \hskip 0.3in}r@{ (}r@{) \hskip 0.3in}r@{ (}r@{) \hskip 0.3in}r@{ (}r@{) \hskip 0.3in}}
   &\multicolumn{2}{c}{\itshape gic} &  \multicolumn{2}{c}{\itshape cv}& \multicolumn{2}{c}{\itshape aic}& \multicolumn{2}{c}{\itshape bic} & \multicolumn{2}{c}{\itshape gbic} & \multicolumn{2}{c}{\itshape kl oracle}\\ 
   \hline
 &\multicolumn{12}{c}{}\\
 p&\multicolumn{12}{c}{\textbf{Specificity}}\\
 &\multicolumn{12}{c}{}\\
    & N/A & N/A & N/A & N/A& N/A & N/A& N/A&N/A & N/A &N/A & N/A & N/A\\ 
  & N/A & N/A & N/A & N/A& N/A & N/A& N/A&N/A & N/A &N/A & N/A & N/A\\ 
   & N/A & N/A & N/A & N/A& N/A & N/A& N/A&N/A & N/A &N/A & N/A & N/A\\ 
    & N/A & N/A & N/A & N/A& N/A & N/A& N/A&N/A & N/A &N/A & N/A & N/A\\    &\multicolumn{12}{c}{}\\
   &\multicolumn{12}{c}{\textbf{Sensitivity}}\\
   &\multicolumn{12}{c}{}\\
   30  & 29.70 & 3.19 & 37.16 & 5.26 & $\bold{56.15}$ & 3.70 & 10.12 & 5.57 & 9.01 & 4.36 & 32.77 & 3.31 \\ 
   60  & 20.74 & 2.16 & 22.55 & 3.09 & $\bold{48.35}$ & 2.59 & 1.69 & 1.26 & 1.95 & 1.12 & 22.59 & 1.46 \\ 
   90  & 16.39 & 2.07 & 16.91 & 2.54 & $\bold{42.69}$ & 2.63 & 0.72 & 0.59 & 1.39 & 0.83 & 18.48 & 1.09 \\ 
   120  & 14.20 & 1.61 & 14.03 & 2.26 & $\bold{42.66}$ & 1.81 & 0.27 & 0.15 & 0.42 & 0.30 & 14.50 & 0.66 \\ 
           &\multicolumn{12}{c}{}\\
   &\multicolumn{12}{c}{\textbf{MCC}}\\
   &\multicolumn{12}{c}{}\\
    & N/A & N/A & N/A & N/A& N/A & N/A& N/A&N/A & N/A &N/A & N/A & N/A\\ 
  & N/A & N/A & N/A & N/A& N/A & N/A& N/A&N/A & N/A &N/A & N/A & N/A\\ 
   & N/A & N/A & N/A & N/A& N/A & N/A& N/A&N/A & N/A &N/A & N/A & N/A\\ 
    & N/A & N/A & N/A & N/A& N/A & N/A& N/A&N/A & N/A &N/A & N/A & N/A\\ &\multicolumn{12}{c}{}\\
    \hline
 \end{tabular}
 \caption{Comparison of average (SE) support recovery for MODEL 3 which is a DENSE RANDOM GRAPH where SCAD has been used to estimate the precision matrix 
 over 100 replications \label{tab:bandgraph_2} }
 \end{table}
 \end{landscape}

\bibliographystyle{model1-num-names}
\bibliography{References}







\end{document}